\documentclass[accepted]{uai2026} % after acceptance, for a revised version; 
\usepackage[american]{babel}
\usepackage[round]{natbib} % has a nice set of citation styles and commands
\usepackage{mathtools} % amsmath with fixes and additions
\usepackage{booktabs} % commands to create good-looking tables
\usepackage{tikz} % nice language for creating drawings and diagrams

\usepackage{bm}
\usepackage{bbm}
\usepackage{amsmath}
\usepackage{amsthm}
\usepackage{mathtools}
\usepackage{amssymb}
\newtheorem{definition}{Definition}[section]
\newtheorem{remark}{Remark}[section]

\newtheorem{proposition}{Proposition}[section]

\newtheorem{lemma}{Lemma}[section]
\usepackage{algorithm}
\usepackage{algorithmic}
\usepackage{caption}
\usepackage{subcaption}

\title{Multiproposal Elliptical Slice Sampling}

\author[1,2]{\href{guillermina.senn@ntnu.no}{Guillermina Senn}{}}
\author[2]{Nathan Glatt-Holtz}
\author[2]{Giulia Carigi}
\author[3]{Andrew Holbrook}
\author[1]{Håkon Tjelmeland}
\affil[1]{%
    Dept. of Mathematical Sciences\\
    NTNU\\
    Trondheim, Norway
}
\affil[2]{%
    Dept. of Statistics\\
    Indiana University Bloomington\\
    USA
}
\affil[3]{%
    Dept. of Biostatistics\\
    Fielding School of Public Health, UCLA\\
    USA
  }
  
\begin{document}
\maketitle

\begin{abstract}
We introduce Multiproposal Elliptical Slice Sampling, a self-tuning multiproposal Markov chain Monte Carlo method for Bayesian inference with Gaussian priors. 
Our method generalizes the Elliptical Slice Sampling algorithm by 1) allowing multiple candidate proposals to be sampled in parallel at each self-tuning step, and 2) basing the acceptance step on a distance-informed transition matrix that can favor proposals far from the current state.
This allows larger moves in state space and faster self-tuning, at essentially no additional wall clock time for expensive likelihoods, and results in improved mixing.  
We additionally provide theoretical arguments and experimental results suggesting dimension-robust mixing behavior, making the algorithm particularly well suited for Bayesian PDE inverse problems.
\end{abstract}

\section{Introduction}\label{sec:introduction}
Markov chain Monte Carlo (MCMC) algorithms are the workhorse of Bayesian inference for sampling from a general posterior distribution. A reason for their success is the theoretical guarantee that the algorithm will generate samples from the exact posterior distribution in the limit of infinitely many iterations. To keep the number of iterations reasonably small in practice, the algorithm should be efficient in converging to and exploring the target distribution.

Both the convergence and exploration behavior of many MCMC methods depend on algorithmic parameters. To achieve acceptable efficiency, an appropriate value for these parameters must be found, a procedure referred to as tuning. Tuning is in general an algorithm-specific, time-consuming step that often requires multiple preliminary runs. 
Adaptive MCMC \citep{Rosenthal2011} offers a feasible alternative to tuning for some algorithms, based on the premise of diminishing adaptations, which is proven to not alter the desired stationary distribution. However, the adaptation might be slow or inefficient in practice. 
Tuning-free MCMC algorithms are therefore extremely convenient for the practitioner. 

A famous example of a tuning-free MCMC algorithm is the No-U-Turn-Sampler \citep{Hoffman2014}, that automates the choice of trajectory length in Hamiltonian Monte Carlo (HMC; \citet{Duane1987}). Despite of NUTS now allowing a tuning-free efficient exploration of a high-dimensional general posterior distribution, gradients still need to be computed.  A viable alternative for the class of Bayesian posteriors with Gaussian priors is embodied by the gradient-free, self-tuning Elliptical Slice Sampling algorithm from \citet{Murray2010}. 

Elliptical Slice Sampling (ESS) is  defined for target distributions of the form
% $\pi(\bm{x}),\bm{x}\in\mathbb{R}^n$ defined by
\begin{align} \label{eq:target}
\pi(\bm{x}) \propto L(\bm{x}) \, \mathcal{N}(\bm{x}; \bm{\mu}, \bm{\Sigma}), \quad \bm{x}\in\mathbb{R}^n,
\end{align}
where $L(\bm{x})$ is a non-negative likelihood function. The prior-informed proposal in ESS resembles that of \citet{Neal1999} and the later infinite-dimensional generalization in \citet{Cotter2013}, the so-called pre-conditioned Crank-Nicholson (pCN) algorithm, and is therefore also defined as a linear combination of the current state and a sample from the Gaussian prior, where the weights of the linear combination govern the aggressiveness of the proposal. 
The ESS proposal improves upon the previous, first by allowing sampling from the full ellipse defined by the current state and the prior draw, and second by using Slice Sampling \citep{Neal2003} to automatically determine appropriate weights for the linear combination.

What appears to be previously unnoted about ESS, and extends to our generalizations herein, is that the pCN-like structure of the proposals on the ellipse strongly suggest that this algorithm can be made direct sense of in an infinite-dimensional setting. This means in practice that we can resolve~\eqref{eq:target} without suffering from the curse of dimensionality. 
We will substantiate these claims further in Section~\ref{sec:computational} and with a non-parametric numerical example in Section~\ref{sec:results}. 

\begin{algorithm}[t]
\caption{Elliptical Slice Sampler (modified)}
\label{alg:ess}
\begin{algorithmic}[1]
\REQUIRE Current state $\bm{x}$, sample from the prior $\bm{\nu}$, likelihood $L(\bm{x})$. 
\STATE Sample auxiliary variable $u \sim \mathcal{U}(0, 1)$ to define slice level $y \gets u \cdot L(\bm{x})$.
\STATE Sample $\alpha \sim \mathcal{U}[0, 2\pi)$ .
\STATE Initialize shrinking bounds $l_0(\alpha) \gets 0$, $r_0(\alpha) \gets 2\pi$.
\REPEAT
    \STATE Sample $\varphi_i \sim 
    \mathcal{U}\big(l_{i-1}(\alpha,\bm{\varphi}_{1:i-1}),
    r_{i-1}(\alpha,\bm{\varphi}_{1:i-1})\big]$.
    \STATE Propose $\tilde{\bm{x}} =
    \bm{x}\cos(\varphi_i-\alpha) +
    \bm{\nu}\sin(\varphi_i-\alpha)$.
    \IF{$L(\tilde{\bm{x}}) > y$}
        \STATE Accept $\bm{x} \gets \tilde{\bm{x}}$
    \ELSE
        \IF{$\varphi_i < \alpha$}
            \STATE Shrink left bound: $l_{i}(\alpha, \bm{\varphi}_{1:i}) \gets \varphi_i$, and keep $r_{i}(\alpha, \bm{\varphi}_{1:i}) = r_{i-1}(\alpha, \bm{\varphi}_{1:i-1})$.
        \ELSE
            \STATE Shrink right bound: $r_{i}(\alpha, \varphi_{1:i}) \gets \bm{\varphi}_i$, and keep $\ell_{i}(\alpha, \bm{\varphi}_{1:i}) = \ell_{i-1}(\alpha, \bm{\varphi}_{1:i-1})$.
        \ENDIF
    \ENDIF
\UNTIL{acceptance}
\end{algorithmic}
\end{algorithm}

In this article, we consider a slightly modified version of the original ESS algorithm. We present the method in Algorithm~\ref{alg:ess} and describe it next in words.
Letting $\bm{x}$ denote the current state, we generate a new state in two steps. First, independently of $\bm{x}$, a $\bm{\nu}\sim \mathcal{N}(\bm{\mu}, \bm{\Sigma})$ is generated from the Gaussian prior, and an angle $\alpha\sim\mbox{Unif}(0,2\pi]$ is sampled. Together with $\bm{x}$, the sample $\bm{\nu}$ defines an ellipse centered at $\bm{\mu}\in\mathbb{R}^n$, that can be parametrically described by
\begin{equation}\label{eq:ellipse}
\begin{split}
    \widetilde{\bm{x}} &= (\bm{x}-\bm{\mu})\cos (\theta - \alpha) + (\bm{\nu}-\bm{\mu}) \sin (\theta - \alpha) + \bm{\mu}, \\
    \widetilde{\bm{\nu}} &= (\bm{\nu}-\bm{\mu})\cos (\theta - \alpha) - (\bm{x}-\bm{\mu}) \sin (\theta - \alpha) + \bm{\mu},
\end{split}
\end{equation}
for $\theta \in (0, 2\pi]$.
The angle $\alpha$ determines the position of  $\bm{x}$ on the ellipse and is located $\alpha$ radians anticlockwise from 0.
This redefinition of the algorithm achieves a natural repositioning of the initial length-$2\pi$ interval around the angle associated to the current state. As a consequence, shrinking occurs immediately after the first rejected angle, and simplifies the proof of invariance.
Next, an angle $\theta$ is sampled by a procedure to be described, and the resulting $\widetilde{\bm{x}}$ is the next state of the Markov chain.

The angle $\theta$ is drawn through the following sequence of operations. First, a scalar $y$ is generated from a uniform distribution on $(0,L(\bm{x})]$, defining a \textit{slice} $\mathcal{S}:=\{\bm{x}:L(\bm{x}) \geq y\}$ (see \citet{Neal2003} for an extended explanation on slice sampling). 
Next, an iterative procedure is performed until an acceptance criterion is met. Letting $i$ denote the iterator index, and starting with $i=1$, we sample an angle $\varphi_i$ uniformly on the interval $(\ell_{i-1}(\alpha, \bm{\varphi}_{1:i-1}),r_{i-1}(\alpha, \bm{\varphi}_{1:i-1})]$, with initial conditions $\ell_{0}(\alpha) = 0$ and $r_{0}(\alpha)=2\pi$, and evaluate the likelihood in~\eqref{eq:target} in the $\widetilde{\bm{x}}$ defined by $\theta=\varphi_i$ in~\eqref{eq:ellipse}. If $L(\widetilde{\bm{x}})\geq y$, we set $\theta=\varphi_i$ and finish. If $L(\widetilde{\bm{x}})< y$, a shrunken interval $(\ell_i(\alpha, \bm{\varphi}_{1:i}),r_i(\alpha, \bm{\varphi}_{1:i)}]$ is defined according to the ESS shrinking rule in Algorithm~\ref{alg:ess}.
Then, the value of the iterator index $i$ is incremented and the next iteration starts. The procedure continues until the generated $\varphi_i$ leads to a proposal $\widetilde{\bm{x}}$ that falls on the slice. For continuous likelihoods, the algorithm is guaranteed to terminate because $\alpha$ is always within the interval. 

In this article, we introduce Multiproposal Elliptical Slice Sampling (MESS), a multiproposal extension of ESS that allows $M$ angles to be sampled in parallel at each iteration. MESS can enjoy improved efficiency over ESS ($M=1$) by the following three mechanisms. First, searching for a slice with $M>1$ simultaneous angles in the same interval increases the probability of finding a segment of the slice far away from the current state, as illustrated in Figure~\ref{fig:ellipses}(a). Second, using all $M$ angles to shrink a bracket results in faster shrinking, as illustrated in Figure~\ref{fig:ellipses} (b-c). We note that both first and second are obtained at no additional cost per iteration compared to the case $M=1$ when parallelization is used and the likelihood cost dominates over the parallelization overhead. Third, the acceptance step can be designed based on a distance-informed transition matrix to favor proposals far from the current state.

The rest of the article is organized as follows. In Section~\ref{sec:mess} we describe MESS and give a mathematically precise proof of invariance, that becomes a proof for ESS by setting $M=1$. In Section~\ref{sec:transition} we explain how to construct the transition matrix in the acceptance step. In Section~\ref{sec:computational} we describe how MESS scales computationally and statistically. In Section~\ref{sec:related} we review the use of parallelization and multiple proposals in the ESS literature. In Section~\ref{sec:results} we test MESS on the logistic regression studied in \citet{Murray2010} and on two Bayesian inverse problems: the blind image deconvolution problem with a multimodal posterior from \citet{Senn2026}, and a non-parametric Bayesian inverse problem inspired by the sparse indirect measurement of turbulent fluid flow. We conclude and give directions of future research in Section~\ref{sec:conclusion}.

\section{Multiproposal Elliptical Slice Sampling}\label{sec:mess}
In the following we first describe in detail the updating procedure used in MESS, and thereafter give a mathematical precise proof that $\pi(\bm{x})$ is an invariant distribution for the Markov chain defined by this updating procedure.

\subsection{The updating procedure}\label{sec:updating}
\begin{algorithm}[t]
\caption{Multiproposal Elliptical Slice Sampler}
\label{alg:mess}
\begin{algorithmic}[1]
\REQUIRE Current state $\bm{x}$, sample from the prior $\bm{\nu}$, likelihood $L(\bm{x})$, number of proposals $M$. 
\STATE Sample auxiliary variable $u \sim \mathcal{U}(0, 1)$ to define slice level $y \gets u \cdot L(\bm{x})$.
\STATE Sample $\alpha \sim \mathcal{U}(0, 2\pi]$ to reposition the initial interval randomly around the current state.
\STATE Initialize shrinking bounds $l_0(\alpha) \gets 0$, $r_0(\alpha) \gets 2\pi$.
\REPEAT
    \STATE Sample $\varphi_{ij} \sim \mathcal{U}(l_{i-1}(\alpha, \bm{\varphi}_{1:i-1}), r_{i-1}(\alpha, \bm{\varphi}_{1:i-1})]$
    \STATE Propose $\tilde{\bm{x}}^{(j)} = \bm{x} \cos (\varphi_{ij} - \alpha) + \bm{\nu} \sin(\varphi_{ij} - \alpha)$ for $j=1,\ldots,M$.
    \IF{$L(\tilde{\bm{x}}^{(j)}) < y$ for $j=1, \ldots, M$}
        \STATE Redefine $l_{i}(\alpha, \bm{\varphi}_{1:i})$ and $r_{i}(\alpha, \bm{\varphi}_{1:i})$ using the shrinking rule in Definition~\ref{def:shrinking_rule}.
    \ELSE
        \STATE Randomly choose $m$ from $\{j: L(\tilde{\bm{x}}^{(j)}) > y\}.$
        \STATE Accept $\bm{x} \gets \tilde{\bm{x}}^{(m)}$ and \textbf{return}
    \ENDIF
    \UNTIL{acceptance}
\end{algorithmic}
\end{algorithm}
The MESS pseudocode is given in Algorithm~\ref{alg:mess}.
Given $\bm{x}$, $\bm{\nu}$, $\alpha$ and $y$ generated as in Algorithm~\ref{alg:ess}, the following  procedure is used to sample the angle $\theta$. Starting with $i=1$, a vector $\bm{\varphi}_i = (\bm{\varphi}_{i1}, \ldots, \bm{\varphi}_{iM})$ of $M$ angles is sampled independently and uniformly on the interval $(\ell_{i-1}(\alpha, \bm{\varphi}_{1:i-1}),r_{i-1}(\alpha, \bm{\varphi}_{1:i-1})]$, with initial conditions $\ell_{0}(\alpha) = 0$ and $r_{0}(\alpha)=2\pi$. 
Next, we check which (if any) of the $M$ angles leads to a proposal with likelihood above the threshold $y$. 
To perform this verification, it is useful to define the set $\mathcal{A}_i$ with the indexes of the \textit{valid} angles sampled at iteration $i$, 
\begin{equation}\label{eq:A}
\begin{split}
\mathcal{A}_i
:=\; &\mathcal{A}(\bm{x}, \bm{\nu}, y, \alpha, \bm{\varphi}_i)  \\
=\; &\bigl\{ j :
L\!\bigl(\bm{x}\cos(\varphi_{ij}-\alpha)
      + \bm{\nu}\sin(\varphi_{ij}-\alpha)\bigr) \ge y, \\
&\qquad j = 1,\ldots,M \bigr\}.
\end{split}
\raisetag{6pt}
\end{equation}

If no valid angles are found, i.e. $|\mathcal{A}_i|=0$, a shrunken interval $(\ell_i(\alpha, \bm{\varphi}_{1:i}),r_i(\alpha, \bm{\varphi}_{1:i)}]$ is defined according to the shrinking rule described below.
\begin{definition}[MESS Shrinking rule]\label{def:shrinking_rule}
Let $\alpha \in (0, 2\pi]$. For some positive $k$, the shrinking rule for 
$i=1,2, \ldots$ and $j=1,\ldots,M$ is phrased as the recursive updates
\begin{subequations}\label{eq:update_sh}
\begin{align}
\hspace*{-0.1cm}
\ell_i(\alpha, \bm{\varphi}_{1:i}) 
&= \max\{\ell_{i-1}(\alpha, \bm{\varphi}_{1:i-1}), 
\varphi_{ij}: \varphi_{ij} < \alpha \},
\label{eq:lUpdate_sh}\\
\hspace*{-0.1cm}
r_i(\alpha, \bm{\varphi}_{1:i}) 
&= \min\{r_{i-1}(\alpha, \bm{\varphi}_{1:i-1}), 
\varphi_{ij}: \varphi_{ij} \geq \alpha \},
\label{eq:rUpdate_sh}
\end{align}
\end{subequations}
of the intervals, together with the initial conditions 
$\ell_0(\alpha,\{\}) = \ell_0(\alpha)=0$ and 
$r_0(\alpha,\{\})=r_0(\alpha)=2\pi$, 
where $\bm{\varphi}_{1:0} = \{\}$, 
$\bm{\varphi}_{1:i} = \{\bm{\varphi}_1, \ldots, \bm{\varphi}_i\}$.
\end{definition}

With this rule, the current interval is shrunken to the smallest interval around $\alpha$ defined by the two (rejected) angles in $\bm{\varphi}_{i}$ that are closest to $\alpha$ from the left and from the right. If no angle in $\bm{\varphi}_{i}$ is smaller than $\alpha$, the current left limit of the interval is unchanged, and analogously for the right limit. We note that setting $M=1$ in Definition~\ref{def:shrinking_rule} gives the ESS shrinking rule. 

\begin{figure}[t]
    \centering
    \includegraphics[width=0.95\linewidth, trim={0.5cm 0cm 0.2cm 0.2cm}, clip]{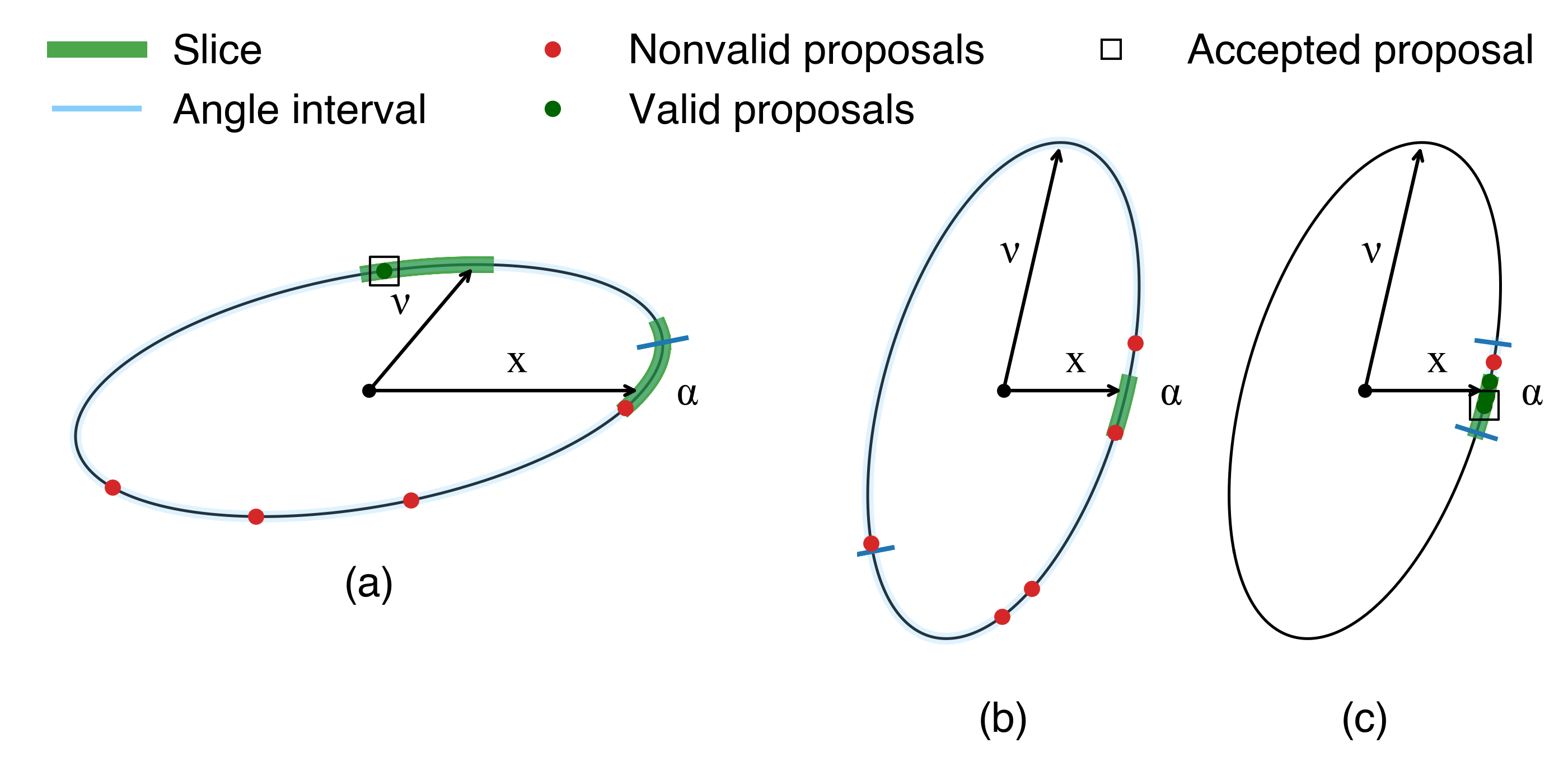}
    \caption{
    Ellipses at different iterations of a MESS ($M=5$) chain for posterior estimation of the non-parametric example in Section~\ref{sec:results}. (a) The slice is discontinuous, and one of five proposals lies on a segment excluding the current state. (b, c) None of the first five angles is valid, shrinking the interval; four subsequent proposals are valid, and one is accepted.}
    \label{fig:ellipses}
\end{figure}

If at least one valid angle is found at iteration $i$, we denote by $k$ the iteration number at which the stop criterion $|\mathcal{A}_k| > 0$ is met, define the vector 
\begin{equation}\label{eq:valid_angles}
    \bm{\varphi}_{k, \mathcal{A}_{k}} = (\varphi_{kj}:j\in \mathcal{A}_k)
\end{equation}
with the valid angles, and choose one among them using probabilities specified by a transition matrix whose construction we describe next.

The transition matrix has to include probabilities for going from the current state to any valid angle sampled at iteration $k$. Therefore we define the augmented vector $\bm{\psi} = (\alpha, \bm{\varphi}_{k, \mathcal{A}_k})$ with all $B = |\mathcal{A}_k| + 1$ valid angles. The initial angle $\alpha$ takes the position 0 in this new vector. 
% \textcolor{blue}{The first element of $\bm{\psi}$ is defined to be the angle corresponding to the current state.} 
Next, and using $^{\uparrow}$ to denote that a vector has been sorted in increasing order of magnitude, we define
\begin{equation*}
    \bm{\psi}^{\uparrow} := (\psi_{(1)}, \ldots, \psi_{(B)}), \quad \psi_{(1)} \le \psi_{(2)} \le \cdots \le \psi_{(B)},
\end{equation*}
the sorted vector of valid angles, and construct the following transition matrix for the elements of $\bm{\psi}^{\uparrow}$. 
\begin{definition}[Transition matrix] \label{def:transition_matrix}
The transition matrix for the elements of the ordered vector $\bm{\bm{\psi}}^{\uparrow}$ of candidate angles at iteration $k$ is a double-stochastic matrix
\begin{equation*}
    \bm{P}(\bm{\psi}^{\uparrow}) \;=\;\left[P_{rs}(\bm{\psi}^{\uparrow})\right]_{r,s=1}^B
\end{equation*}
of order $B$ whose entries $P_{rs}(\bm{\psi}^{\uparrow})$ are the probabilities for transitioning from angle $\bm{\psi}_{(r)}$ to angle $\bm{\psi}_{(s)}$. The matrix is restricted to have diagonal entries $P_{rs}(\bm{\psi}^{\uparrow})=0$ to forbid staying in the current state. 
\end{definition}
\begin{remark}[Invariance]\label{rem:invariance}
    The ordered vector $\bm{\psi}^{\uparrow}$ and therefore $\bm{P}(\bm{\psi}^\uparrow)$ are invariant to permutations of the elements of $\bm{\psi}$.
\end{remark}

To specify this last step of the algorithm we introduce notation for the mapping of the elements of $\bm{\psi}$ to the elements of $\bm{\psi}^\uparrow$, and vice versa.
We define the ordering function $h$ that returns the position that each angle in $\bm{\psi}$ has in the sorted vector $\bm{\psi}^\uparrow$, i.e.
\begin{equation}
    h(\psi_r; \bm{\psi})=
        \text{position of $\psi_{r}$ in $\bm{\psi}^{\uparrow}$},
\end{equation}
and correspondingly the inverse function $h^{-1}$ that returns the label that $\psi_{(r)}$ has in the unsorted $\bm{\psi}$,
\begin{equation}
    h^{-1}(\psi_{(r)}; \bm{\psi}^{\uparrow})= \text{label of } \psi_{(r)}.
\end{equation}
Using this notation, the acceptance procedure is as follows. First, obtain the position $r=h(\alpha; \bm{\psi})$ of $\alpha$ in $\bm{\psi}^{\uparrow}$, then sample an integer $s \in \{1, \ldots, B\}$ according to the probabilities in row $r$ of the transition matrix, and finally recover the index $m=h^{-1}(\psi_{(s)}; \bm{\psi}^{\uparrow})$ of the accepted angle, thus denoted as $\varphi_{km}$. Then set $\theta=\varphi_{km}$ and finish. 

Having described the updating procedure in MESS, we formulate the proposition below about the chain's invariant distribution. We give a summary of its proof in Section~\ref{sec:validity} and more details in Appendix~\ref{ap:proof_mess}.
\begin{proposition}[Invariance] \label{thm:stationary_mess} 
    The Multiproposal Elliptical Slice Sampling updating procedure leaves the target distribution $\pi(\bm{x})$ invariant. 
\end{proposition}

\subsection{Validity of the updating procedure}\label{sec:validity}
To show that $\pi(\bm{x})$ is an invariant distribution for the Markov chain defined by the MESS procedure, we assume the current state $\bm{x}$ to be distributed according to the target distribution and show that the updating procedure keeps the marginal distribution for $\bm{x}$ invariant. 

The strategy of the proof is as follows.
Given the current state $\bm{x}$, the algorithm samples values for the random variables $\bm{\phi}=(\bm{\nu}, y, \alpha, k, \bm{\varphi}_{1},\ldots,\bm{\varphi}_{k})$ and $m$ generated in the updating procedure, with support specified by the constraint set
\begin{equation} \label{eq:support}
\begin{aligned}
\mathcal{F}
= \Bigl\{ (\bm{x},\bm{\phi}, m) :
&\ \bm{x}, \bm{\nu} \in \mathbb{R}^n,\;
   y \in (0,L(\bm{x})], \; k>0, \\
&\ \alpha \in (0,2\pi],\;
   \varphi_{ij} \in (\ell_{i-1}, r_{i-1}], \\
&\ |\mathcal{A}_i|=0\ \forall i<k,\;
   m \in \mathcal{A}_k
\Bigr\}.
\end{aligned}
\end{equation}

In the first step of the proof, we find the density of the realized forward path. Then, we write a one-to-one transformation $T$ that maps $(\bm{x}, \bm{\phi}, m)$ to $(\widetilde{\bm{x}}, \widetilde{\bm{\phi}}, \widetilde{m})$, where 
\begin{align}
    \widetilde{\bm{\phi}}=(\widetilde{\bm{\nu}},\widetilde{y}, \widetilde{\alpha}, \widetilde{k}, \widetilde{\bm{\varphi}}_1,\ldots,\widetilde{\bm{\varphi}}_{\widetilde{k}-1}, \widetilde{\bm{\varphi}}_{\widetilde{k}}) 
\end{align} 
and $\widetilde{m}$ form a set of values that gives a corresponding inverse move, i.e. a move starting with $\widetilde{\bm{x}}$ as the current state and ending up with $\bm{x}$ as the updated state.
Next, we use the transformation formula to find the joint density for the realized reverse path. In this step we employ two lemmas (formulated below) that describe properties of the support in~\eqref{eq:support} and of the intervals defined by the shrinking procedure in Definition~\ref{def:shrinking_rule}.
Finally, we sum the forward and reverse joint densities over $m$ and $\widetilde{m}$ respectively, and use the double-stochasticity of the transition matrix to show that $\bm{x}$ and $\widetilde{\bm{x}}$ are marginally identically distributed. 

We now describe the deterministic transformation $T:(\bm{x},\bm{\phi}, m)\to(\widetilde{\bm{x}},\widetilde{\bm{\phi}}, \widetilde{m})$. The slice level must be the same in both paths, therefore $\widetilde{y} = y$. To have a one-to-one transformation, we must use $\widetilde{k}=k$, and it is natural to let the accepted (extended) state $(\widetilde{\bm{x}}, \widetilde{\bm{\nu}})$ and the angles in the reverse path be given as
\begin{equation}\label{eq:transformation}
\begin{split}
\widetilde{\bm{x}} &= \bm{x}\cos(\varphi_{km} - \alpha) + \bm{\nu}\sin(\varphi_{km} - \alpha),\\
\widetilde{\bm{\nu}} &= \bm{\nu}\cos(\varphi_{km} - \alpha) -\bm{x}\sin(\varphi_{km} - \alpha)\\
\widetilde{\alpha} &= \varphi_{km},\\
\widetilde{\varphi}_{ij} &= \varphi_{ij} \mbox{~for $i=1,\ldots,k-1, \; j=1,\ldots,M$},\\
\widetilde{\varphi}_{kj} &= \varphi_{kj} \mbox{~for $j \ne m$},\\
\widetilde{\varphi}_{\widetilde{k}\widetilde{m}} &= \alpha.
\end{split}
\end{equation}
For the reverse path to be such that the accepted angle is $\alpha$, we must choose $\widetilde{m}=m$.  
The inverse transformation $T^{-1}:(\widetilde{\bm{x}},\widetilde{\bm{\phi}},\widetilde{m})\rightarrow (\bm{x},\bm{\phi},m)$ is then given by
\begin{equation}\label{eq:transformation_inverse}
\begin{split}
\bm{x} &= \widetilde{\bm{x}}\cos(\widetilde{\varphi}_{\widetilde{k}\widetilde{m}} - \widetilde{\alpha}) + \widetilde{\bm{\nu}}\sin(\widetilde{\varphi}_{\widetilde{k}\widetilde{m}} - \widetilde{\alpha}),\\
\bm{\nu} &= \widetilde{\bm{\nu}}\cos(\widetilde{\varphi}_{\widetilde{k}\widetilde{m}} - \widetilde{\alpha}) -\widetilde{\bm{x}}\sin(\widetilde{\varphi}_{\widetilde{k}\widetilde{m}} - \widetilde{\alpha}),\\
\alpha &= \widetilde{\varphi}_{\widetilde{k}\widetilde{m}},\\
\varphi_{ij} &= \widetilde{\varphi}_{ij} \mbox{~for $i=1,\ldots,\widetilde{k}-1, \; j=1,\ldots,M$},\\
\varphi_{kj} &= \widetilde{\varphi}_{\widetilde{k}j} \mbox{~for $j \ne \widetilde{m}$},\\
\varphi_{km} &= \widetilde{\alpha}.
\end{split}
\end{equation}
We denote by $\widetilde{\bm{\psi}}$ the vector of valid angles at iteration $\widetilde{k}$ computed on $T^{-1}(\widetilde{\bm{x}},\widetilde{\bm{\phi}},\widetilde{m})$.
\begin{remark}\label{rem:psi_transformed}
The vector $\widetilde{\bm{\psi}}$ has $\widetilde{\varphi}_{\widetilde{k}\widetilde{m}}$ in position $h^{-1}(\alpha, \bm{\psi}^{\uparrow})=0$, $\widetilde{\alpha}$ in position $h^{-1}(\varphi_{km}, \bm{\psi}^{\uparrow})$, and $\widetilde{\varphi}_{kj}$ in position $h^{-1}(\varphi_{kj}, \bm{\psi}^{\uparrow})$ for the remaining elements. 
    From Remark~\ref{rem:invariance}, it follows that $\widetilde{\bm{\psi}}^{\uparrow}=\bm{\psi}^{\uparrow}$. 
\end{remark}
It is straigthforward to show that the transformation for the continuous variables has unit Jacobian.
Moreover, we propose the lemma below that states that the support in~\eqref{eq:support} is preserved by the transformation $T$. 
\begin{lemma}[Invariance of support] \label{lem:equivalence_indicators}
Under the transformation $T$ described above,
\begin{equation*}
\mathbb{I}\,\!\bigl\{T^{-1}(\widetilde{\bm{x}},\widetilde{\bm{\phi}},\widetilde{m}) \in \mathcal{F}\bigr\}
=
\mathbb{I}\bigl\{(\widetilde{\bm{x}},\widetilde{\bm{\phi}},\widetilde{m})\in \mathcal{F}\bigr\}.
\end{equation*}
\end{lemma}

The proof of Lemma~\ref{lem:equivalence_indicators}, (included in Appendix~\ref{ap:proof_indicators}), uses a second lemma (stated below and proved in Appendix~\ref{ap:proof_recursive_mess}) that formalizes in a mathematically precise manner some important properties of the angle brackets that appear as a consequence of the MESS shrinking rule.
\begin{lemma}[Ordering and flipping properties of intervals] \label{lem:recursive_mess}
    Let $\alpha \in (0, 2\pi]$ and assume the shrinking rule in Definition~\ref{def:shrinking_rule}. Assume further the angle membership property
    \begin{align}\label{eq:membership_mess}
        \varphi_{ij}\in(\ell_{i-1}(\alpha, \bm{\varphi}_{1:i-1}),r_{i-1}(\alpha, \bm{\varphi}_{1:i-1})]
    \end{align}
    for $i=1,\ldots,k$ and $j=1, \ldots, M$. 
    Then, the ordering properties
    \begin{align}\label{eq:lordering_mess}
        \ell_{i-n}(\alpha, \bm{\varphi}_{1:i-n}) \leq \ell_{i}(\alpha, \bm{\varphi}_{1:i})< \alpha,\\
        \alpha \leq r_{i}(\alpha, \bm{\varphi}_{1:i}) \leq r_{i-n}(\alpha, \bm{\varphi}_{1:i-n})\label{eq:rordering_mess}
    \end{align}
    hold for $i=1, \ldots, k-1$ and $n=1, \ldots, i$.
    Moreover, the inequalities
    \begin{equation} \label{eq:intermediate_result_mess}
        \varphi_{ij} < \alpha \iff \varphi_{ij} < \varphi_{km} \;
        \text{and} \; 
        \varphi_{ij} \ge \alpha \iff \varphi_{ij} \ge \varphi_{km}
    \end{equation} 
    hold for $i=1, \ldots, k-1$ and $j, m=1, \ldots, M$, and consequently the flipping properties
    \begin{equation} \label{eq:equivalence_limits_mess}
        \ell_{i}(\alpha, \bm{\varphi}_{1:i}) = \ell_{i}(\varphi_{km}, \bm{\varphi}_{1:i}) 
    \; \text{and} \;
    r_{i}(\alpha, \bm{\varphi}_{1:i}) = r_{i}(\varphi_{km}, \bm{\varphi}_{1:i})
    \end{equation}
    hold for $i = 0, \ldots, k-1$.
\end{lemma}

We now have all elements needed to identify the joint densities of the forward and backward paths. Namely, the joint forward density becomes
\begin{equation} \label{eq:joint_before}
\begin{aligned}
p_{\bm{x},\bm{\phi}, m}(\bm{x},\bm{\phi}, m)
&= g_{\bm{x},\bm{\phi}}(\bm{x},\bm{\phi})
   P_{h(\alpha; \bm{\psi}), h(\varphi_{km}; \bm{\psi})}(\bm{\psi}^{\uparrow})  \\
&\quad \times \mathbb{I}\bigl((\bm{x}, \bm{\phi}, m)\in \mathcal{F}\bigr),
\end{aligned}
\end{equation}
where
\begin{equation} \label{eq:partial_joint_before}
\begin{split}
    g_{\bm{x},\bm{\phi}}(\bm{x},\bm{\phi}) &\propto \frac{\mathcal{N}(\bm{x}; \bm{\mu}, \bm{\Sigma})\;\mathcal{N}(\bm{\nu}; \bm{\mu}, \bm{\Sigma})}{2\pi}  \\
&\times \prod_{i=1}^{k} \prod_{j=1}^{M} 
   \frac{1}{r_{i-1}(\alpha,\bm{\varphi}_{1:i-1})-\ell_{i-1}(\alpha,\bm{\varphi}_{1:i-1})}.
\end{split}
\raisetag{46pt}
\end{equation}
Now, recalling that the transformation in~\eqref{eq:transformation} has unit Jacobian, the joint reverse density is given from the transformation formula as
\begin{equation}\label{eq:transformation_formula}
    p_{\widetilde{\bm{x}}, \widetilde{\bm{\phi}}, \widetilde{m}}(\widetilde{\bm{x}}, \widetilde{\bm{\phi}}, \widetilde{m}) = p_{\bm{x}, \bm{\phi}, m}(T^{-1}(\widetilde{\bm{x}}, \widetilde{\bm{\phi}}, \widetilde{m})).
\end{equation}
Next, as detailed in Appendix~\ref{ap:proof_mess}, we use   Lemma~\ref{lem:equivalence_indicators}, Remark~\ref{rem:invariance}, and Lemma~\ref{lem:recursive_mess}, to rewrite the RHS of~\eqref{eq:transformation_formula} as
\begin{equation} \label{eq:joint_after}
\begin{split}
p_{\bm{x}, \bm{\phi}, m}(T^{-1}(\widetilde{\bm{x}}, \widetilde{\bm{\phi}}, \widetilde{m})) &= 
    p_{\bm{x},\bm{\phi}}(\widetilde{\bm{x}},\widetilde{\bm{\phi}}) P_{h(\widetilde{\varphi}_{\widetilde{k}\widetilde{m}}; \widetilde{\bm{\psi}}), h(\widetilde{\alpha}; \widetilde{\bm{\psi}})}(\widetilde{\bm{\psi}}^{\uparrow}) \\
    &\times \mathbb{I}((\widetilde{\bm{x}},\widetilde{\bm{\phi}} ,\widetilde{m})\in \mathcal{F}).
\end{split}
\end{equation} 

Summing over $m$ in~\eqref{eq:joint_before} and over $\widetilde{m}$ in~\eqref{eq:joint_after} we obtain from the double-stochastic property of the transition matrix that $(\bm{x}, \bm{\phi})$ and $(\widetilde{\bm{x}}, \widetilde{\bm{\phi}})$ are identically distributed.
As the marginal distribution for $\bm{x}$ by construction is the target distribution $\pi(\bm{x})$, the updated state $\widetilde{\bm{x}}$ must also be distributed according to the target distribution. Thus, the target distribution is an invariant distribution for the Markov chain defined by MESS. We employ the same irreducibility and aperiodicity arguments given by \citet{Murray2010} for ESS to claim that the chain converges to the invariant distribution. 

\section{Transition Matrix}\label{sec:transition}
A simple choice for choosing the elements of the transition matrix $\bm{P}(\bm{\psi}^{\uparrow})$ in Definition~\ref{def:transition_matrix} is simply to set all elements, except those in the diagonal, equal to $1/(B-1)$. This is equivalent to accepting uniformly among the $B-1$ valid sampled angles. 
A more sophisticated alternative, used e.g. in \citet{tjelmeland2004using}, is to specity a model for the probabilities in the transition matrix through a set of constraints that the elements in the matrix have to fulfill, and then find values for the elements subject to these constraints.

Following this line of thought, we ask from the transition matrix that: 

1) it is a valid transition matrix for the Markov chain, 
\begin{equation} \label{eq:valid_transition_matrix}
\begin{split}
    0 \leq P_{rs}(\bm{\psi}^{\uparrow}) \leq 1 \hspace{0.5em} \text{for } r, s = 0, \ldots, B-1\\
    \sum_{s=0}^{B-1} P_{rs}(\bm{\psi}^{\uparrow}) = 1 \hspace{0.5em} \text{for } r = 0, \ldots, B-1,
\end{split}   
\end{equation}
2) it is double stochastic, 
\begin{equation} \label{eq:double_stochastic}
\begin{split}
    \sum_{r=0}^{B-1} P_{rs}(\bm{\psi}^{\uparrow}) = 1 \hspace{0.5em} \text{for } s = 0, \ldots, B-1,
\end{split}   
\end{equation}
and that 3) self-transitions are not allowed,
\begin{equation} \label{eq:no_self_transitions}
\begin{split}
    P_{rr}(\bm{\psi}^{\uparrow})=0 \hspace{0.5em} \text{for } r = 0, \ldots, B-1.
\end{split}   
\end{equation}

To optionally reward larger moves of the MESS Markov chain, we define a distance function $d(r, s) \geq 0$ that somehow quantifies the difference between a pair of angles $(\psi_{(r)}, \psi_{(s)})$, for $0 \leq r,s \leq B-1$. 
The entries of the transition matrix can then thus defined as the solution to an optimization problem that involves maximizing the objective function 
\begin{equation} \label{eq:objective_function}
    \sum_{r=0}^{B-1}\sum_{s=0}^{B-1} d(r,s) P_{rs}(\bm{\psi}^{\uparrow})
\end{equation}
subject to the constraints in~\eqref{eq:valid_transition_matrix}, \eqref{eq:double_stochastic}, and~\eqref{eq:no_self_transitions}.
The optimization problem in~\eqref{eq:objective_function} is linear in the unknown entries of the matrix and therefore can be solved with linear programming (LP) using standard solution algorithms. See e.g. \citet{BoydVandenberghe2004} for an introduction to linear programming.

With respect to the choice of distance function, in this article we explore the great-circle angular distance between two angles $\psi_{(r)}$ and $\psi_{(s)}$,
\begin{equation} \label{eq:angular_dist}
    d(r, s) = \min(|\psi_{(r)} - \psi_{(s)}|, 2\pi - |\psi_{(r)} - \psi_{(s)}|),
\end{equation}
and the Euclidean distance in state space, although other choices are allowed. 

\section{Computational and statistical scaling} \label{sec:computational}
Each MESS step consists of $k$ subiterations with $M$ likelihood evaluations each, plus ordering $\bm{\psi}$ and solving the LP problem. Thus MESS has a cost 
\begin{equation}
    \mathcal{O}(kM\cdot\mbox{cost of lik} + \mbox{cost LP})
\end{equation} per iteration. 
Parallelizing the proposal step and likelihood evaluations to $Z$ independent processors reduces the complexity to $\mathcal{O}(kM\cdot\mbox{cost of lik}/Z + \mbox{cost LP})$, without considering the overhead. The LP problem has cost $\mathcal{O}(B^2)$ when all entries of the transition matrix must be computed, and thus might dominate the total cost if $B$ is large and the likelihood evaluation is cheap. We leave our ideas to reduce the cost of the optimization problem for a future article.

In a body of recent works, e.g. \citet{BeskosPinskiSanzSernaStuart2011, Cotter2013, GlattHoltz2024}, a family of new algorithms adapted to sampling~\eqref{eq:target} in a dimension-free fashion was proposed. 
The main idea is to take proposal structures that respect the form of the prior measures in a way that the method is defined directly on Hilbert spaces.
Because MESS (and therefore ESS) have a pCN-type proposal structure, which holds the prior as invariant, we assert here that MESS falls under this class of dimension-free algorithms.

In forthcoming work we will aim to justify this claim on a rigorous foundation drawing upon the weak Harris theorem, as in
\citet{HairerStuartVollmer2014} and our more recent work in \citet{Carigi2026} (to appear). 
As it stands, we have compelling numerical evidence for this dimension-free behaviour in MESS. See particularly the results for the solute transport toy model featuring a non-Gaussian, non-parametric posterior in Section~\ref{sec:results} below.

\section{Related parallel-ESS work}\label{sec:related}
In the ESS literature, the Generalized Elliptical Slice Sampling (GESS; \citet{nishihara2014parallel}) and the Transport Elliptical Slice Sampling (TESS; \citet{cabezas2023transport} aim to improve the correspondence of the prior with the posterior, a feat achieved by proposing a scale-mixture of Gaussians approximation to the posterior in the first, and via normalizing flows in the second. The resulting satisfactory mixing is further improved in both algorithms by the use of parallel computing to run two communicating groups of parallel chains.

To the best of our knowledge, parallelization has not been used in the ESS literature to parallelize the proposal mechanism and likelihood evaluation. However, in Elliptical Slice Sampling with Expectation Propagation (EPESS; \citet{fagan2016elliptical}) the authors introduce \textit{Recycled ESS}, where each ESS cycle generates sequentially $M$ valid proposals instead of just one. This sequential generation causes the angular distance between successive proposals to decrease in average, probably explaining the small improvement in computational and statistical efficiency observed by the authors. 

\section{Results}\label{sec:results}
We first compare MESS to ESS for the Gaussian process classification example studied in \citet{Murray2010}, and then demonstrate MESS variants on two Bayesian inverse problems. The code and data for all the examples is available online at \href{https://github.com/guillerminasenn/mess}{https://github.com/guillerminasenn/mess}.

\subsection{Gaussian process classification}
See \citet{Murray2010} for a description of the problem. 
Figure~\ref{fig:shrinking_steps} shows how, for this problem, increasing the number of proposals geometrically reduces the mean number of shrinking steps per iteration. The mean number of likelihood iterations, estimated as the product $M \times \text{mean nr. of shrinking steps}$, increases linearly with $M$, and so does the computational cost per iteration without considering parallelization. If parallelization with $Z$ cores is used, the computational time per iteration is proportional to $M/Z \times \text{mean nr. of shrinking steps}$.
\begin{figure}
    \centering
    \includegraphics[width=0.9\linewidth, trim={0cm 0 0 0cm, clip}]{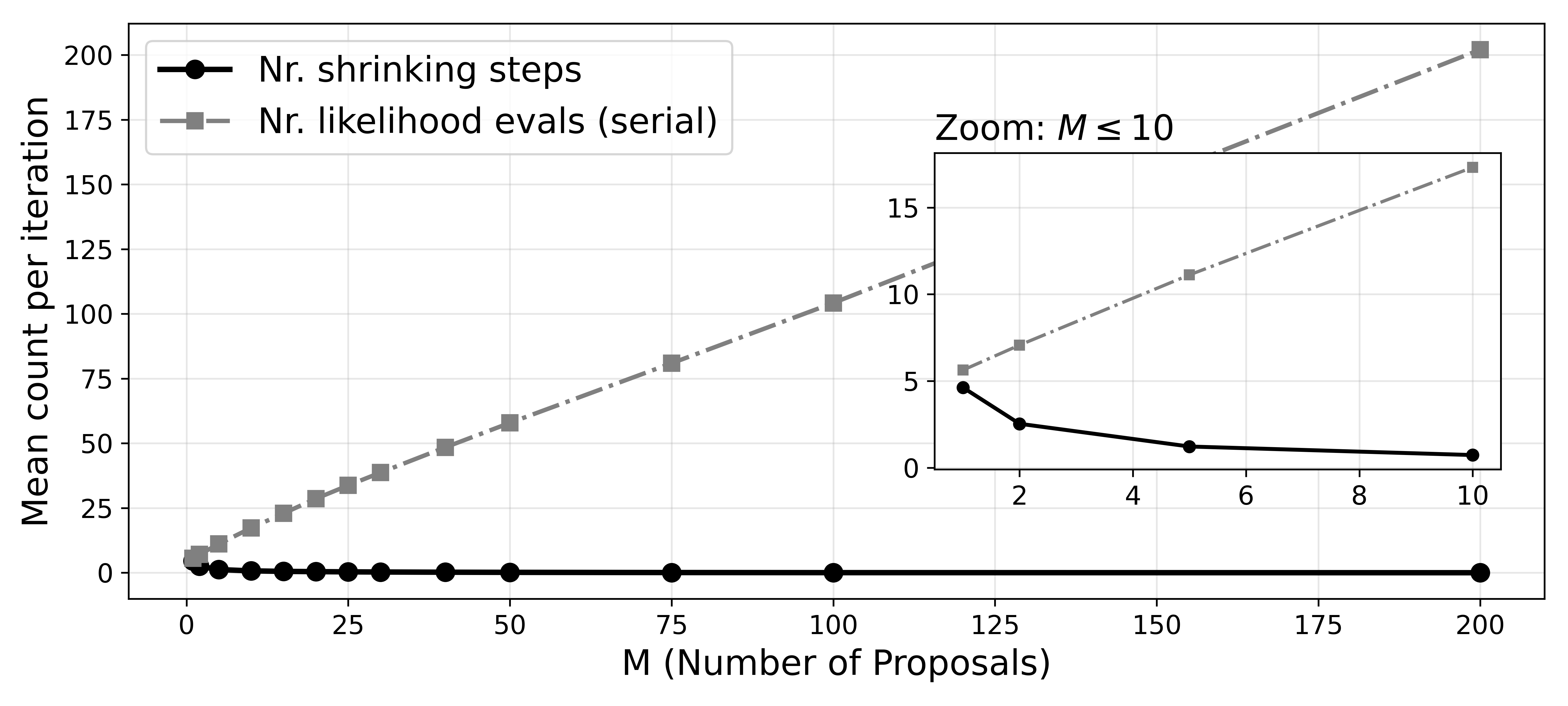}
    \caption{Average number of shrinking steps and likelihood evaluations per iteration, as a function of the number of proposals, for the Gaussian process classification example.}
    \label{fig:shrinking_steps}
\end{figure}

\subsection{Blind deconvolution}\label{sec:bd}
Blind deconvolution (BD) is an ill-posed inverse problem where the goal is to simultaneously recover the true signal and the unknown convolution or blur kernel from a set of observations corrupted by blur and noise. The one-dimensional BD model is mathematically formulated as
\begin{equation} \label{eq:bd}
    \bm{d} = \bm{w} \star \bm{c} + \bm{e}, 
\end{equation}
where $\bm{d},\bm{c},\bm{e}\in\mathbb{R}^n$ are the vectorized observations, signal and noise, respectively, and $\bm{w}\in\mathbb{R}^k$ is the 1D blur kernel. \citet{BulandOmre2003b,Senn2025,Senn2026} approach the simultaneous estimation of $\bm{w}$ and $\bm{c}$ from~\eqref{eq:bd} as a Bayesian inverse problem with Gaussian likelihood and Gaussian (independent) signal and blur priors. The intractable posterior exhibits perfect sign-shift multimodality when both priors have zero mean, due to $\bm{w}\star\bm{c}=-\bm{w}\star(-\bm{c})$. 
This posterior structure is especially amenable to MESS since the proposal always considers the negative of the current state.

We simulate data from~\eqref{eq:bd} using $n=24$ and $k=12$, and constrain the signal prior with 4 exact signal observations, thus shifting the prior mean away from zero but not enough to destroy all posterior multimodality. 
We then compare the behaviour of MESS (with $M=20$ and angular distance) with the benchmark Gibbs sampler and HMC algorithms from \citet{Senn2025, Senn2026}, for the marginal estimation of $\bm{w}$. As evidenced by the traceplots and posterior samples in Figure~\ref{fig:bd},  HMC and MESS find both sign-shift modes in the first 50k iterations, while Gibbs finds only one. The gradient-free MESS needed 6min for these iterations, compared to the 2.5hrs required by HMC. 

\begin{figure}
    \centering
    \includegraphics[width=1\linewidth, trim={2cm 1cm 2cm 0.9cm}, clip]{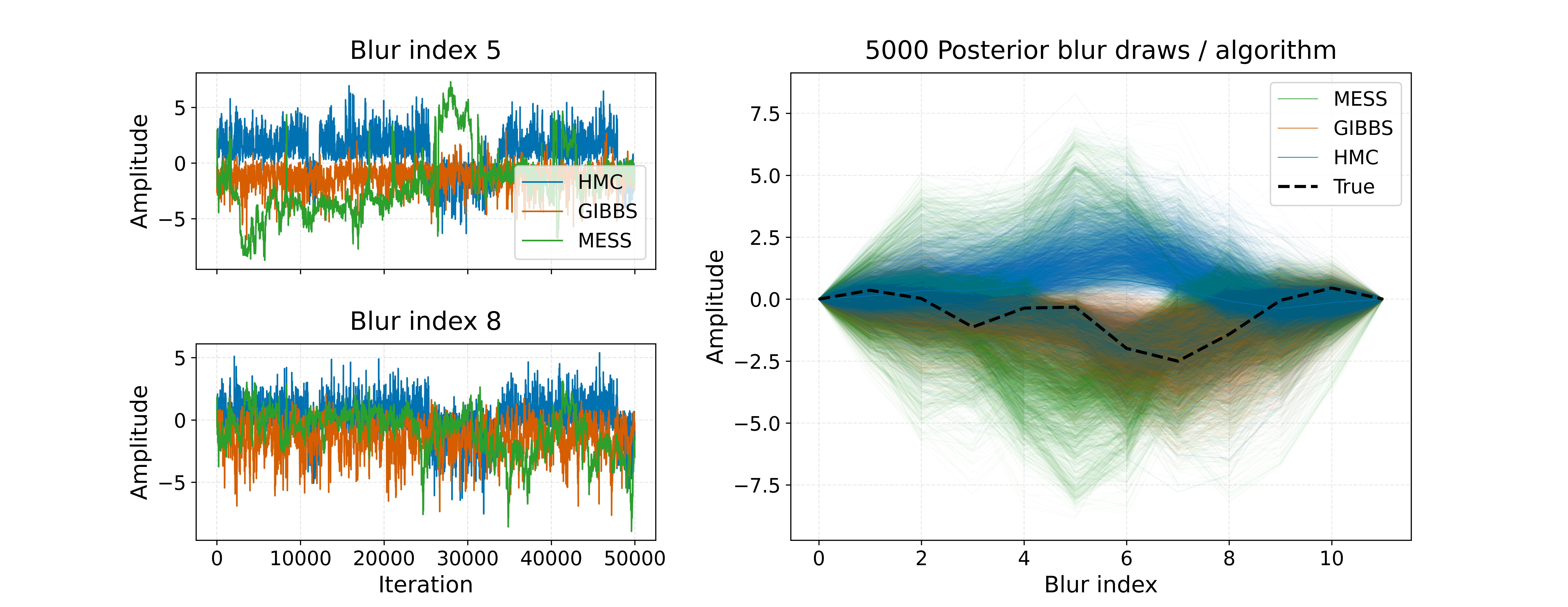}
    \caption{Posterior samples and selected traceplots for the marginal blur estimation in the blind deconvolution model. MESS ($M=20$, angular) explores both modes in 4\% of the time used by HMC.}
    \label{fig:bd}
\end{figure}

\subsection{Toy model for solute transport}\label{sec:ad}
We consider the following toy model \citep{GlattHoltz2024} which mimics some of the features of the transportation of a dye (or any other solute) by a fluid under damping and external forcing: 
\begin{equation}\label{eq:ad_toy}
    (\bm{A} + \kappa \bm{I}) \bm{\Theta} = \bm{g}.
\end{equation}
Here, $\bm{\Theta} \in \mathbb{R}^d$ is the state variable (e.g. dye concentration), $\kappa>0$ is a damping parameter, $\bm{g} \in \mathbb{R}^d$ is a time-independent forcing term (e.g. stirring action), and $\bm{A}$ is a zero-diagonal matrix of order $d$ that plays the role of a discretized bounded operator. The model is already written in a basis that behaves like Fourier modes, so that the elements $a_{ij}$ represents the energy transfer from mode $i$ to $j$. To model a natural physical symmetry, $\bm{A}$ is defined to be antisymmetric and thus specified by the the $m = d(d-1)/2$ non-zero elements in its upper triangle. 
The model in~\eqref{eq:ad_toy} extends to infinite dimensions assuming appropriate square-summability of the coefficients of $\bm{A}$ and $\bm{g}$ and hence $\bm{\Theta}$. 

Our goal is to use Bayesian inversion on~\eqref{eq:ad_toy} to estimate $\bm{A}$. For the estimation we assume the partial observational model
\begin{equation}\label{eq:ad_obs_model}
    \bm{y} = \mathcal{P}(\bm{\Theta}(\bm{A})) + \bm{\varepsilon}, \quad \bm{\varepsilon} \sim N_k(0, \sigma^2 \bm{I}_k),
\end{equation}
where $\mathcal{P}:\mathbb{R}^d \to \mathbb{R}^k$ is a projection operation and $\sigma^2$ is the observational noise scale. The estimation problem can be phrased for finite $d$ as sampling from a posterior with the form in ~\eqref{eq:target} 
with $L(\bm{A}) := \exp{(-\Phi(\bm{A}))}$, $\Phi(\bm{A})$ being a quadratic potential, and prior $N(\bm{0}, \bm{C})$. The covariance matrix $\bm{C}$ is chosen to enforce a local advection operator. See Appendix~\ref{ap:toy_ad} for a more detailed problem specification. 

\begin{figure}
    \centering
    \includegraphics[width=0.85 \linewidth, trim={0.5cm 0 0 0}, clip]{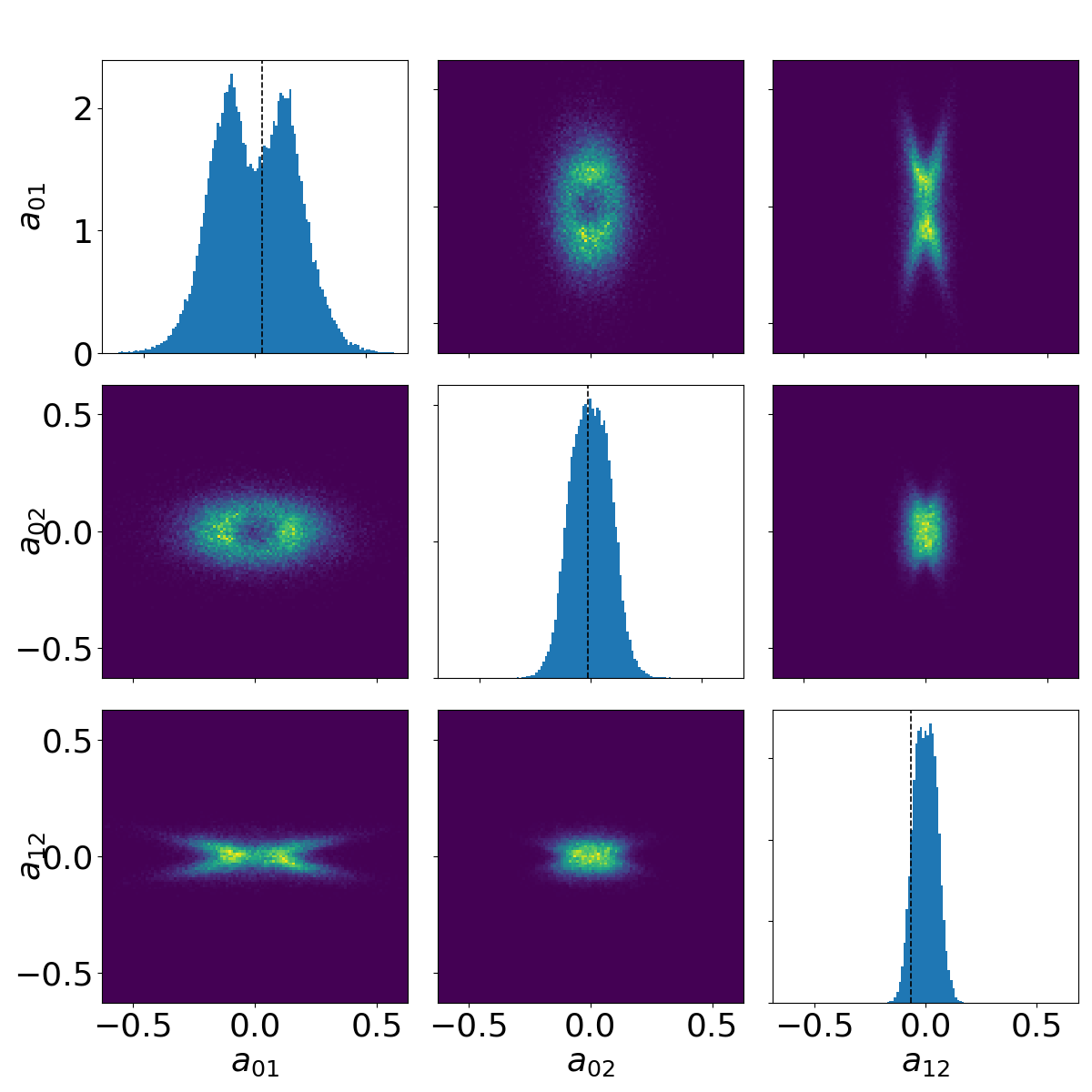}
    \caption{Posterior marginal histograms for $a_{01}, a_{02}$, and $a_{12}$ on the diagonal, and corresponding pairwise density plots off-diagonal, for the solute transport model at $d=10$, computed with 300k samples from MESS $(M=100)$. The highly non-Gaussian marginal posteriors and complex correlation structures are a result of the non-linear forward map.}
    \label{fig:ad_pairplots_n3}
\end{figure}

\begin{figure}
    \centering
    \includegraphics[width=0.9\linewidth, trim={0cm 0 0 0, clip}]{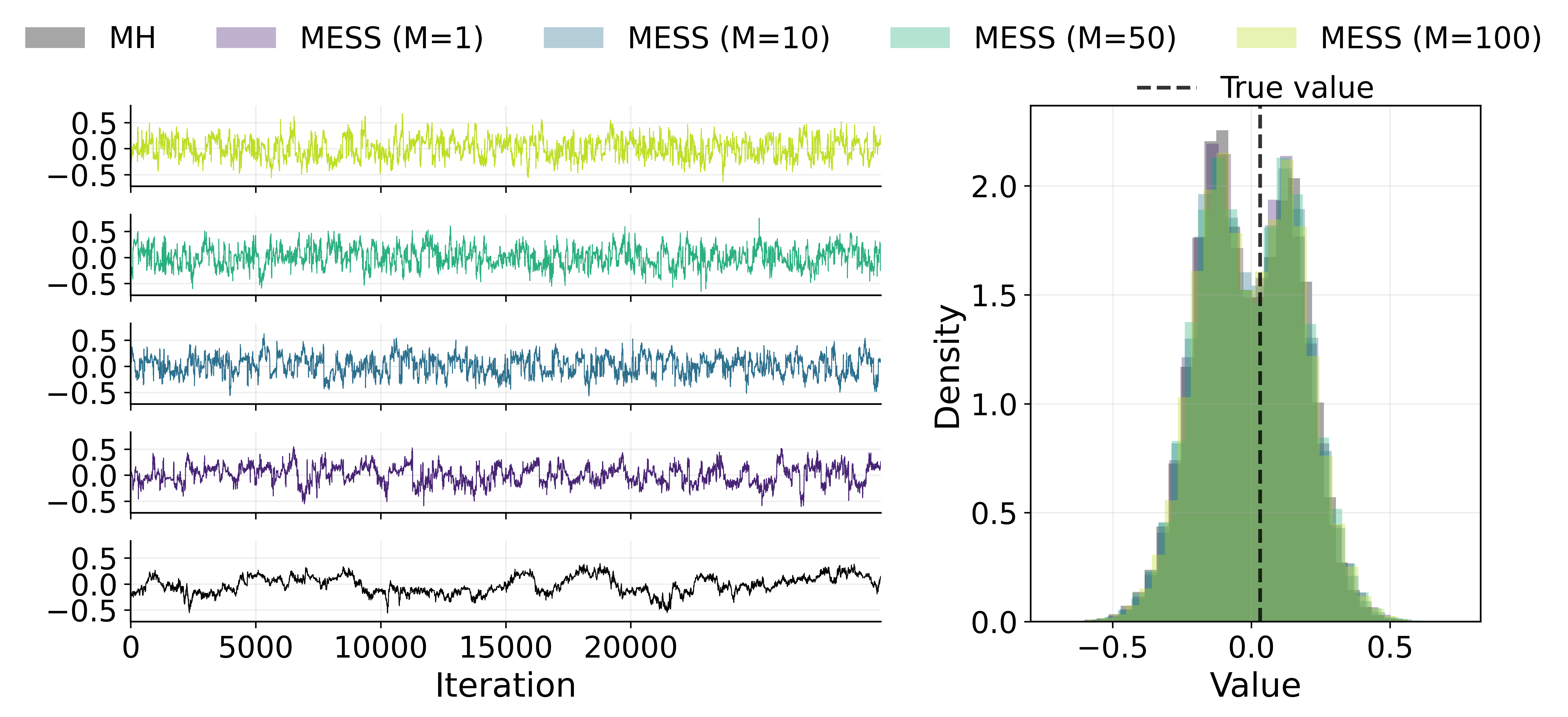}
    \caption{Posterior traceplots (30k iterations) and marginal histogram (300k iterations) for $a_{01}$ in the solute transport model, at $d=10$. All chains satisfactorily explore the posterior although with different exploration behaviour.}
    \label{fig:traceplots_ad}
\end{figure}

We perform a numerical experiment to assess the dimension-robust mixing behavior of MESS. We generate datasets for $d\in\{10, 15, \ldots, 50\}$ as  described in Appendix~\ref{ap:toy_ad}, and for each dataset estimate the corresponding $d \times d$ matrix $\bm{A}_d$ using the same observational scale and single long MESS and Metropolis-Hastings (MH) chains. For MESS, we consider $M\in \{1,10,50,100\}$, and for MH, the proposal is tuned to achieve give 23.4\% acceptance rate at $d=20$. 
We fix the observational scale and the MH proposal standard deviation across dimensions to prevent that the mixing of a chain is compensated for, for increased dimension, by increased information or better tuning.

Figure~\ref{fig:ad_pairplots_n3} illustrates the complex posterior geometry obtained at $d=10$ and using samples from MESS $(M=100)$. The visualization includes marginal histograms for $a_{01}, a_{02}$, and $a_{12}$ on the diagonal, and corresponding pairwise density plots off-diagonal. The highly non-Gaussian marginal posteriors and complex correlation structures are a result of the non-linear map in~\eqref{eq:ad_toy}. 

Figure~\ref{fig:traceplots_ad} (right) shows how all MESS variantes and MH explore the marginal for $a_{01}$ satisfactorily at $d=10$ for the 300k considered iterations, although with different exploration behaviour, as illustrated by the first 10k iterations plotted in Figure~\ref{fig:traceplots_ad} (left). 

\begin{figure}
    \centering
    \includegraphics[width=1\linewidth, trim={0 0.2cm 0 0cm}, clip]{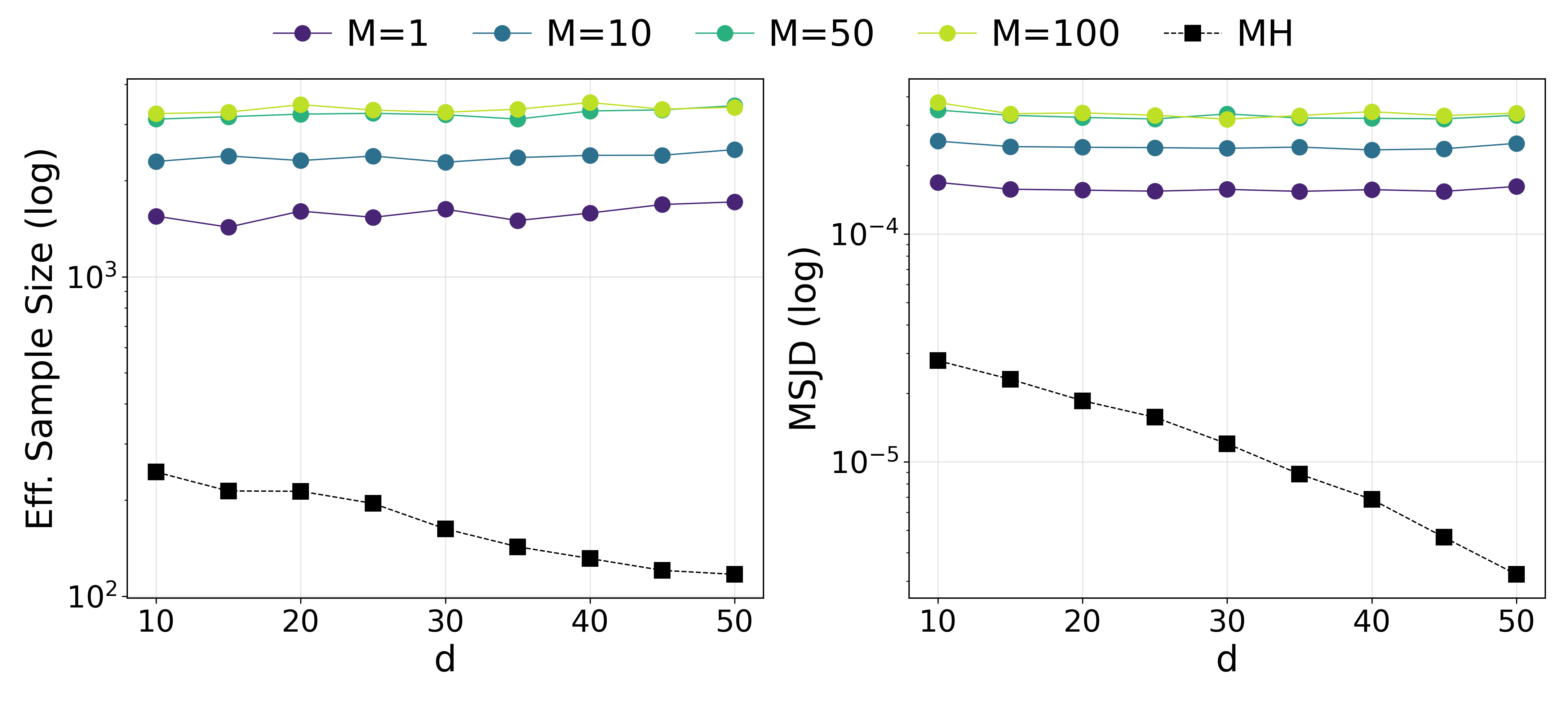}
    \caption{Effective sample size and MSJD in log-scale as a function of dimension, for several values of $M$ and the benchmark Metropolis-Hastings, for the solute transport model. The mixing rate of MESS does not degrade with dimension.}
    \label{fig:mixing_ad}
\end{figure}

Figure~\ref{fig:mixing_ad} gives empirical indication that the mixing rate (as measured by the effective sample size and the MSJD) does not degrade with dimension for MESS, as it happens with the benchmark MH. The same figure shows that the mixing improves with $M$ in MESS, up to a point where the number of proposals saturates the ellipse. 

Table~\ref{tab:ess_msjd_mean_d10} compares mixing performance at $d=10$ using uniform, angular, and Euclidean MESS chains. Informing the transition matrix in Section~\ref{sec:transition} with Euclidean distance slightly improved the average effective sample size, where for the average we considered the 8 elements $\{a_{01}, a_{02}, a_{03}, a_{04}, a_{12}, a_{13}, a_{23}, a_{24}\}$ of $\bm{A}_{10}$ that exhibited the most non-Gaussian marginal distributions. The ESS for Euclidean distance measured individually per component yielded improvements of up 20\% for some components (see Appendix~\ref{ap:toy_ad}, Table~\ref{tab:ess_msjd_ad_a02}). 
\begin{table}
\centering
\small
\begin{tabular}{l|ccc|ccc}
\toprule
& \multicolumn{3}{c|}{Mean Eff. Sample Size} & \multicolumn{3}{c}{MSJD} \\
$M$ & Unif & Ang & Eucl & Unif & Ang & Eucl \\
\midrule
10 & 2295 & 2241 & \textbf{2334} & 0.00026 & 0.00025 & 0.00026 \\
50 & 3113 & 3061 & \textbf{3149} & 0.00035 & 0.00034 & 0.00035 \\
\bottomrule
\end{tabular}
\caption{Average ESS/MSJD obtained with different MESS variants at $d=10$ in the solute transport model. The values are obtaines with 300k samples and averaged over eight elements of $\bm{A}_{10}$ with non-Gaussian marginals. Euclidean MESS performs slightly better on average than the other variants.}
\label{tab:ess_msjd_mean_d10}
\end{table}

\section{Concluding remarks}\label{sec:conclusion}
We proposed Multiproposal Elliptical Slice Sampling (MESS), a generalization of the Elliptical Slice Sampling algorithm for posterior inference in models with Gaussian priors, where $M$ angles are sampled in parallel and the acceptance step can favor proposals far from the current state. 
Experimental results showed improved mixing driven by a higher probability of reaching slice segments far from the current state when $M>1$, faster bracket shrinking, and the use of a distance-informed transition matrix in the acceptance step.
In the non-parametric numerical example, the mixing rate did not degrade with dimension, suggesting dimension-robust mixing behavior.

Future work involves establishing theoretical results for the dimension-free claim, along with studying the interaction between the number of proposals and the geometries induced by different likelihoods, and reducing the computational cost of the optimization problem used to compute the transition matrix. 

\begin{acknowledgements} % will be removed in pdf for initial submission,
						 % (without ‘accepted’ option in \documentclass)
                         % so you can already fill it to test with the
                         % ‘accepted’ class option
    This work was funded by the GAMES and CGF projects at NTNU, and by NSF DMS-2236854, DMS-2108790, and DMS-1816551.
\end{acknowledgements}

% References
\bibliography{references}

\newpage

\onecolumn

\title{Multiproposal Elliptical Slice Sampling\\(Supplementary Material)}
\maketitle

\appendix
% Reset and prefix equation numbering
\renewcommand{\theequation}{A.\arabic{equation}}
\setcounter{equation}{0}

% Reset and prefix figure numbering
\renewcommand{\thefigure}{A.\arabic{figure}}
\setcounter{figure}{0}

% Reset and prefix table numbering
\renewcommand{\thetable}{A.\arabic{table}}
\setcounter{table}{0}

\section{Proof of Lemma~\ref{lem:equivalence_indicators}}\label{ap:proof_indicators}
\begin{proof}[Proof of Lemma~\ref{lem:equivalence_indicators}]

Using the transformation in~\eqref{eq:transformation_inverse}, the indicator function $\mathbb{I}(T^{-1}(\widetilde{\bm{x}}, \widetilde{\bm{\phi}}, \widetilde{m}) \in \mathcal{F})$
equals one if and only if the following expressions are fulfilled,
\begin{equation}
\label{eq:rest1}
\begin{aligned}
(a)\;\; & \widetilde{\varphi}_{\widetilde{k}\widetilde{m}} \in \bigl(0,\, 2\pi\bigr], \\[2pt]
(b)\;\; & \widetilde{\varphi}_{ij} \in 
    \bigl(
        \ell_{i-1}\bigl(\widetilde{\varphi}_{\widetilde{k}\widetilde{m}},\, \widetilde{\bm{\varphi}}_{1:i-1}\bigr),\;
        r_{i-1}\bigl(\widetilde{\varphi}_{\widetilde{k}\widetilde{m}},\, \widetilde{\bm{\varphi}}_{1:i-1}\bigr)
    \bigr]
    \quad \text{for } i=1,\ldots,\widetilde{k}-1,\; j=1,\ldots,M, \\[2pt]
(c)\;\; & \widetilde{\varphi}_{\widetilde{k}j} \in 
    \bigl(
        \ell_{\widetilde{k}-1}\bigl(\widetilde{\varphi}_{\widetilde{k}\widetilde{m}},\, \widetilde{\bm{\varphi}}_{1:\widetilde{k}-1}\bigr),\;
        r_{\widetilde{k}-1}\bigl(\widetilde{\varphi}_{\widetilde{k}\widetilde{m}},\, \widetilde{\bm{\varphi}}_{1:\widetilde{k}-1}\bigr)
    \bigr]
    \quad \text{for } j \neq \widetilde{m}, \\[2pt]
(d)\;\; & \widetilde{\alpha} \in 
    \bigl(
        \ell_{\widetilde{k}-1}\bigl(\widetilde{\varphi}_{\widetilde{k}\widetilde{m}},\, \widetilde{\bm{\varphi}}_{1:\widetilde{k}-1}\bigr),\;
        r_{\widetilde{k}-1}\bigl(\widetilde{\varphi}_{\widetilde{k}\widetilde{m}},\, \widetilde{\bm{\varphi}}_{1:\widetilde{k}-1}\bigr)
    \bigr], \\[2pt]
(e)\;\; & \widetilde{y} \in 
    \bigl(
        0,\;
        L\bigl(
            \widetilde{\bm{x}} \cos(\widetilde{\varphi}_{\widetilde{k}\widetilde{m}} - \widetilde{\alpha})
            + \widetilde{\bm{\nu}} \sin(\widetilde{\varphi}_{\widetilde{k}\widetilde{m}} - \widetilde{\alpha})
        \bigr)
    \bigr], \\[2pt]
(f)\;\; & 
|\mathcal{A}_i\bigl(
    \widetilde{\bm{x}} \cos(\widetilde{\varphi}_{\widetilde{k}\widetilde{m}} - \widetilde{\alpha})
    + \widetilde{\bm{\nu}} \sin(\widetilde{\varphi}_{\widetilde{k}\widetilde{m}} - \widetilde{\alpha}), \\
& \hspace{2.5em}
    -\widetilde{\bm{x}} \sin(\widetilde{\varphi}_{\widetilde{k}\widetilde{m}} - \widetilde{\alpha})
    + \widetilde{\bm{\nu}} \cos(\widetilde{\varphi}_{\widetilde{k}\widetilde{m}} - \widetilde{\alpha}), \\
& \hspace{2.5em}
    \widetilde{y},\; \widetilde{\bm{\varphi}}_i,\; \widetilde{\varphi}_{\widetilde{k}\widetilde{m}}
\bigr)| = 0
\quad \text{for } i = 1, \dots, \widetilde{k}-1 \\[2pt]
(g)\;\; & 
|\mathcal{A}_{\widetilde{k}}\bigl(
    \widetilde{\bm{x}} \cos(\widetilde{\varphi}_{\widetilde{k}\widetilde{m}} - \widetilde{\alpha})
    + \widetilde{\bm{\nu}} \sin(\widetilde{\varphi}_{\widetilde{k}\widetilde{m}} - \widetilde{\alpha}), \\
& \hspace{2.5em}
    -\widetilde{\bm{x}} \sin(\widetilde{\varphi}_{\widetilde{k}\widetilde{m}} - \widetilde{\alpha})
    + \widetilde{\bm{\nu}} \cos(\widetilde{\varphi}_{\widetilde{k}\widetilde{m}} - \widetilde{\alpha}), \\
& \hspace{2.5em}
    \widetilde{y},\; (\widetilde{\bm{\varphi}}_{\widetilde{k}_{-\widetilde{m}}},\, \widetilde{\alpha}),\; \widetilde{\varphi}_{\widetilde{k}\widetilde{m}}
\bigr)| >0 \\[2pt]
(h)\;\; & \widetilde{m} \in\;  
\mathcal{A}_{\widetilde{k}}\bigl(
    \widetilde{\bm{x}} \cos(\widetilde{\varphi}_{\widetilde{k}\widetilde{m}} - \widetilde{\alpha})
    + \widetilde{\bm{\nu}} \sin(\widetilde{\varphi}_{\widetilde{k}\widetilde{m}} - \widetilde{\alpha}), \\
& \hspace{2.5em}
    -\widetilde{\bm{x}} \sin(\widetilde{\varphi}_{\widetilde{k}\widetilde{m}} - \widetilde{\alpha})
    + \widetilde{\bm{\nu}} \cos(\widetilde{\varphi}_{\widetilde{k}\widetilde{m}} - \widetilde{\alpha}), \\
& \hspace{2.5em}
    \widetilde{y},\; (\widetilde{\bm{\varphi}}_{\widetilde{k}_{-\widetilde{m}}},\, \widetilde{\alpha}),\; \widetilde{\varphi}_{\widetilde{k}\widetilde{m}}
\bigr),
\end{aligned}
\end{equation}
where $(a)-(d)$ are related to angle memberships, $(e)-(g)$ to likelihood levels, and $(h)$ to the chosen proposal index.
In the following, we show that the expressions in~\eqref{eq:rest1} are equivalent to 
\begin{equation}
\label{eq:rest2}
\begin{aligned}
(i)\;& \widetilde{\alpha} \in (0, 2\pi], \\[2pt]
(ii)\;& \widetilde{\varphi}_{ij} \in 
      \bigl(\,\ell_{i-1}(\widetilde{\alpha}, 
      \widetilde{\bm{\varphi}}_{1:i-1}),\,
      r_{i-1}(\widetilde{\alpha}, 
      \widetilde{\bm{\varphi}}_{1:i-1})\,\bigr]
      \quad\text{for } i=1,\ldots,\widetilde{k}, \quad j=1,\ldots,P,\\[2pt]
(iii)\;& \widetilde{y} \in (0, L(\widetilde{\bm{x}})], \\[2pt]
(iv)\;& 
|\mathcal{A}_i\bigl(\widetilde{\bm{x}}, \widetilde{\bm{\nu}}, \widetilde{y}, \widetilde{\bm{\varphi}}_{i}, \widetilde{\alpha}\bigr)|=0 \quad \text{for } i=1, \dots, \widetilde{k}-1,\\[2pt]
(v)\;& 
|\mathcal{A}_{\widetilde{k}}\bigl(\widetilde{\bm{x}}, \widetilde{\bm{\nu}}, \widetilde{y}, \widetilde{\bm{\varphi}}_{\widetilde{k}}, \widetilde{\alpha}\bigr)| > 0,\\[2pt]
(vi)\;& \widetilde{m} \in \mathcal{A}_{\widetilde{k}}(\widetilde{\bm{x}}, \widetilde{\bm{\nu}}, \widetilde{y}, \widetilde{\bm{\varphi}}_{\widetilde{k}}, \widetilde{\alpha}).
\end{aligned}
\end{equation}
That~\eqref{eq:rest2} is true is equivalent to $\mathbb{I}\bigl\{(\widetilde{\bm{x}}, \widetilde{\bm{\phi}}, \widetilde{m}) \in \mathcal{F}\bigr\}=1$.

\begin{proof}[Proof for~\eqref{eq:rest1}$(a)-(d) \iff$\eqref{eq:rest2}$(i-ii)$]
The statements in~\eqref{eq:rest1}$(a)-(d)$ allow us to use Lemma~\ref{lem:recursive_mess} with the transformations $\alpha=\widetilde{\varphi}_{\widetilde{k}\widetilde{m}}$, $\varphi_{km}=\widetilde{\alpha}$, and $\varphi_{ij}=\widetilde{\varphi}_{ij}$ for the remaining angles, and therefore properties in the Lemma~\ref{lem:recursive_mess} apply for the intervals and angles defined by the transformation. 

From the ordering properties in Lemma~\ref{lem:recursive_mess}, 
\begin{equation*}
    0 = \ell_{0}(\widetilde{\varphi}_{\widetilde{k}\widetilde{m}}) < \ell_{\widetilde{k}-1}(\widetilde{\varphi}_{\widetilde{k}\widetilde{m}}, \widetilde{\bm{\varphi}}_{1:\widetilde{k}-1}) \quad \text{and} \quad r_{\widetilde{k}-1}(\widetilde{\varphi}_{\widetilde{k}\widetilde{m}}, \widetilde{\bm{\varphi}}_{1:\widetilde{k}-1}) \leq r_{0}(\widetilde{\varphi}_{\widetilde{k}\widetilde{m}}) \leq 2\pi,
\end{equation*}
and then~\eqref{eq:rest1}$(d)$ implies~\eqref{eq:rest2}$(i)$.

From the flipping properties in Lemma~\ref{lem:recursive_mess}, \eqref{eq:rest1}$(b)$ automatically implies \eqref{eq:rest2}$(ii)$ for $i=1, \ldots, \widetilde{k}-1$, and \eqref{eq:rest1}$(c)$ automatically implies \eqref{eq:rest2}$(ii)$ for $i=\widetilde{k}$ with $j\neq \widetilde{m}$. 
To get~\eqref{eq:rest2}$(ii)$ for $j=\widetilde{m}$, we use the ordering properties to note that
\begin{equation*}
    \ell_{\widetilde{k}-1}(\widetilde{\varphi}_{\widetilde{k}\widetilde{m}}, \widetilde{\bm{\varphi}}_{1:\widetilde{k}-1}) < \widetilde{\varphi}_{\widetilde{k}\widetilde{m}} \leq r_{\widetilde{k}-1}(\widetilde{\varphi}_{\widetilde{k}\widetilde{m}}, \widetilde{\bm{\varphi}}_{1:\widetilde{k}-1}),
\end{equation*}
and then we use the flipping properties again.
The proof in the only if direction is given by a similar reasoning.
\end{proof}

\sloppy \begin{proof}[Proof for (\ref{eq:rest1})$(e)-(h) \iff$ (\ref{eq:rest2})$(iii-vi)$]
We first show the if direction. Expression~\eqref{eq:rest2}$(iv)$ is automatically satisfied by using the formula for the sum of two angles on the definition of the set $\mathcal{A}_i$ in~\eqref{eq:A} to rewrite~\eqref{eq:rest1}$(f)$ as
\begin{equation*}
\begin{aligned}
\mathcal{A}_i&\bigl(\widetilde{\bm{x}}\cos(\widetilde{\varphi}_{\widetilde{k}\widetilde{m}} -\widetilde{\alpha})+\widetilde{\bm{\nu}}\sin(\widetilde{\varphi}_{\widetilde{k}\widetilde{m}} - \widetilde{\alpha}), 
    -\widetilde{\bm{x}}\sin(\widetilde{\varphi}_{\widetilde{k}\widetilde{m}} - \widetilde{\alpha}) + \widetilde{\bm{\nu}}\cos(\widetilde{\varphi}_{\widetilde{k}\widetilde{m}} - \widetilde{\alpha}), 
    \widetilde{y}, \widetilde{\bm{\varphi}}_{i}, \widetilde{\varphi}_{\widetilde{k}\widetilde{m}}\bigr) \nonumber \\
    &=\{j:L(\widetilde{\bm{x}}\cos(\widetilde{\varphi}_{ij} - \widetilde{\alpha})+\widetilde{\bm{\nu}}\sin(\widetilde{\varphi}_{ij} - \widetilde{\alpha}))\geq y, \quad j=1,\ldots,P\} \\
    &= \mathcal{A}_i(\widetilde{\bm{x}}, \widetilde{\bm{\nu}}, \widetilde{y}, \widetilde{\bm{\varphi}}_i, \widetilde{\alpha}).
\end{aligned}
\end{equation*}
Similarly, using again the formula for the sum of two angles, we rewrite~\eqref{eq:rest1}$(g)$ as
\sloppy \begin{equation*} 
\begin{aligned}
\mathcal{A}_{\widetilde{k}}&\bigl(\widetilde{\bm{x}}\cos(\widetilde{\varphi}_{\widetilde{k}\widetilde{m}} -\widetilde{\alpha})+\widetilde{\bm{\nu}}\sin(\widetilde{\varphi}_{\widetilde{k}\widetilde{m}} - \widetilde{\alpha}), 
    -\widetilde{\bm{x}}\sin(\widetilde{\varphi}_{\widetilde{k}\widetilde{m}} - \widetilde{\alpha}) + \widetilde{\bm{\nu}}\cos(\widetilde{\varphi}_{\widetilde{k}\widetilde{m}} - \widetilde{\alpha}), 
    \widetilde{y}, (\widetilde{\bm{\varphi}}_{\widetilde{k}_{-\widetilde{m}}}, \widetilde{\alpha}), \widetilde{\varphi}_{\widetilde{k}\widetilde{m}}\bigr) \\
    &= \{j:L(\widetilde{\bm{x}}\cos(\widetilde{\varphi}_{\widetilde{k}j} - \widetilde{\alpha})+\widetilde{\bm{\nu}}\sin(\widetilde{\varphi}_{\widetilde{k}j} - \widetilde{\alpha}))\geq y, \quad j=1,\ldots,P, \quad j \ne \widetilde{m}\} \\
    &\,\cup \{j:L(\widetilde{\bm{x}})\geq y, \quad j=\widetilde{m}\},
\end{aligned}
\end{equation*}
and then we have that (\ref{eq:rest1})$(e)$ and (\ref{eq:rest1})$(g)$ together imply (\ref{eq:rest2})$(v)$. 

Expression~\eqref{eq:rest1}$(h)$ means that $\widetilde{m}$ is the index of an accepted proposal. From the rewritten expression~\eqref{eq:rest1}$(g)$ above we have that $L(\widetilde{\bm{x}}) \geq \widetilde{y}$, and from (\ref{eq:rest1})$(e)$ we know that $\widetilde{y}>0$. Together they give (\ref{eq:rest2})$(iii)$.

Finally, (\ref{eq:rest2})$(v)$ is given because we can build $\mathcal{A}_i(\widetilde{\bm{x}}, \widetilde{\bm{\nu}}, \widetilde{y}, \widetilde{\bm{\varphi}}_i, \widetilde{\alpha})$ from the rewritten expression~\eqref{eq:rest1}$(f)$ above and~\eqref{eq:rest1}$(e)$ and from this last expression, $L(\widetilde{\bm{x}}) \geq \widetilde{y}$. From this, (\ref{eq:rest2})$(vi)$ is satisfied too. The proof for the only if direction is given by a similar reasoning. 
% \textcolor{blue}{Note to self: Verify that this proof is correctly written}
\end{proof}

\end{proof}

\section{Proof of Lemma~\ref{lem:recursive_mess}} \label{ap:proof_recursive_mess}

\begin{proof}[Proof of Lemma~\ref{lem:recursive_mess}]
    Properties~\eqref{eq:lordering_mess} and~\eqref{eq:rordering_mess} are shown by induction. We first prove that~\eqref{eq:lordering_mess} and~\eqref{eq:rordering_mess} are true for all $i$ when $n=1$. 
    For the base case $(i, n)=(1, 1)$, the interval's left endpoint definition~\eqref{eq:lUpdate_sh} gives
    \begin{equation*}
    \ell_{1}(\alpha, \bm{\varphi}_{1}) =
    \begin{cases}
      \max (\bm{\varphi}_{1j}:\bm{\varphi}_{1j}<\alpha), & \text{if } \bm{\varphi}_{1j} < \alpha \text{ for some } j=1, \ldots, M , \\[6pt]
      \ell_{0}(\alpha), & \text{if } \bm{\varphi}_{1j} \ge \alpha \text{ for all } j=1, \ldots, M,
    \end{cases} 
    \end{equation*}
    with $\bm{\varphi}_{1j} > \ell_{0}(\alpha)$ for all $j$ from the angle membership property in~\eqref{eq:membership_mess}. Finally, $\ell_{0}(\alpha) = 0 < \alpha$ by definition. 
    From the interval's right endpoint definition~\eqref{eq:rUpdate_sh} we have
    \begin{equation*}
    r_{1}(\alpha, \bm{\varphi}_{1}) =
    \begin{cases}
      r_{0}(\alpha), & \text{if } \bm{\varphi}_{1j} < \alpha \text{ for all } j=1,\ldots,M, \\[6pt]
      \min(\bm{\varphi}_{1j}) \le r_{0}(\alpha), & \text{if } \bm{\varphi}_{1j} \ge \alpha \text{ for some } j=1,\ldots,M,
    \end{cases}
    \end{equation*}
    with $r_{0}(\alpha) = 2\pi\ge \alpha$ by definition and $\varphi_{1j} \leq r_{0}(\alpha)$ for all $j$ from the angle membership in~\eqref{eq:membership_mess}. 
    
    In the inductive step we assume~\eqref{eq:lordering_mess} and~\eqref{eq:rordering_mess} to be valid for $(i, n) = (k-1, 1)$ and show that then the two relations are valid also for $(i, n) = (k, 1)$. 
    From the interval's left endpoint's definition~\eqref{eq:lUpdate_sh}, we have
    \begin{equation*}
    \ell_{k}(\alpha, \bm{\varphi}_{1:k}) =
    \begin{cases}
     \max(\varphi_{kj}:\varphi_{kj}<\alpha), & \text{if } \varphi_{kj} < \alpha \text{ for some } j=1,\ldots,M, \\[6pt]
     \ell_{k-1}(\alpha, \bm{\varphi}_{1:k-1}) , & \text{if } \varphi_{kj} \ge \alpha \text{ for all } j=1,\ldots,M,
    \end{cases} 
    \end{equation*}
    with $\varphi_{kj} \ge \ell_{k-1}(\alpha, \bm{\varphi}_{1:k-1})$ for all $j$ from the angle membership property~\eqref{eq:membership_mess} and $\ell_{k-1}(\alpha, \bm{\varphi}_{1:k-1}) < \alpha$ by the induction hypothesis. 
    From the interval's right endpoint's definition~\eqref{eq:rUpdate_sh} we have
    \begin{equation*}
    r_{k}(\alpha, \bm{\varphi}_{1:k}) =
    \begin{cases}
      r_{k-1}(\alpha, \bm{\varphi}_{1:k-1}), & \text{if } \varphi_{kj} < \alpha \text{ for all } j=1,\ldots,M, \\[6pt]
      \min(\varphi_{kj}:\varphi_{kj} \geq \alpha) , & \text{if } \varphi_{kj} \ge \alpha \text{ for some } j=1,\ldots,M,
    \end{cases}
    \end{equation*}
    with $\varphi_{kj} \le r_{k-1}(\alpha, \bm{\varphi}_{1:k-1})$ for all $j$ from ~\eqref{eq:membership_mess} 
    and $r_{k-1}(\alpha, \bm{\varphi}_{1:k-1}) \ge \alpha$ by the induction hypothesis. Therefore 
    \begin{align}\label{eq:lordering}
        \ell_{i-1}(\alpha, \bm{\varphi}_{1:i-j}) \leq \ell_{i}(\alpha, \bm{\varphi}_{1:i})< \alpha,\\
        \alpha \leq r_{i}(\alpha, \bm{\varphi}_{1:i}) \leq r_{i-1}(\alpha, \bm{\varphi}_{1:i-1}),\label{eq:rordering}.
    \end{align}
    holds for $i\ge 1$, where the first inequality in (\ref{eq:lordering_mess}) is strict if  $\varphi_{ij} < \alpha$ for at least some $j$.

    We next prove by induction that~\eqref{eq:lordering_mess} is true for a fixed $i$ and $1 \le j \le i$. 
    The base case $(i, j) = (i, 1)$ is given from~\eqref{eq:lordering}. 
    In the inductive step we assume~\eqref{eq:lordering_mess} and~\eqref{eq:rordering_mess} to be valid for $(i, j) = (i, k-1)$ and show that then the two relations are valid also for $(i, j) = (i, k)$.  Assume the inductive hypothesis 
    $$\ell_{i-(k-1)}(\alpha, \bm{\varphi}_{1:i-(k-1)}) = \ell_{i-k+1}(\alpha, \bm{\varphi}_{1:i-k+1}) \leq \ell_{i}(\alpha, \bm{\varphi}_{1:i})< \alpha$$
    for $k \le i$. 
    From~\eqref{eq:lordering} we have
    $$\ell_{i-k}(\alpha, \bm{\varphi}_{1:i-k}) \le \ell_{i-k+1}(\alpha, \bm{\varphi}_{1:i-k+1}) < \alpha$$
    for $i \ge 1$. Combining this property with the inductive hypothesis gives the desired
    $$\ell_{i-k}(\alpha, \bm{\varphi}_{1:i-k}) \le \ell_{i}(\alpha, \bm{\varphi}_{1:i}).$$
    The proof for~\eqref{eq:rordering_mess} is similar.
\end{proof}

\section{Proof of proposition~\ref{thm:stationary_mess}} \label{ap:proof_mess}
Recalling the expression~\eqref{eq:joint_before} for the joint forward density, we show here how it can be used together with the transformation formula to derive expression~\eqref{eq:joint_after} for the joint reverse density. 
First, from the unit Jacobian property of the transformation and the transformation formula, we obtain
\begin{equation}\label{eq:partial_joint_after}
\begin{split}
p_{\widetilde{\bm{x}},\widetilde{\bm{\phi}},\widetilde{m}}(\widetilde{\bm{x}},\widetilde{\bm{\nu}},\widetilde{y}, \widetilde{\alpha},
\widetilde{k}, \widetilde{\bm{\varphi}}_1,\ldots, \widetilde{\bm{\varphi}}_{\widetilde{k}},\widetilde{m}) &
= p_{\bm{x},\bm{\phi},m}(
\widetilde{\bm{x}}\cos(\widetilde{\varphi}_{\widetilde{k}\widetilde{m}} - \widetilde{\alpha})+\widetilde{\bm{\nu}}\sin(\widetilde{\varphi}_{\widetilde{k}\widetilde{m}} - \widetilde{\alpha}), \\
&\hspace{0.9cm} -\widetilde{\bm{x}}\sin(\widetilde{\varphi}_{\widetilde{k}\widetilde{m}} - \widetilde{\alpha})+\widetilde{\bm{\nu}}\cos(\widetilde{\varphi}_{\widetilde{k}\widetilde{m}} - \widetilde{\alpha}),\\
&\hspace{1.4cm}  \widetilde{y},
\widetilde{\varphi}_{\widetilde{k}\widetilde{m}}, \widetilde{k}, \widetilde{\bm{\varphi}}_1,
\ldots,
\widetilde{\bm{\varphi}}_{\widetilde{k}-1}, (\widetilde{\bm{\varphi}}_{{\widetilde{k}}_{-\widetilde{m}}}, \widetilde{\alpha}),\widetilde{m}).
\end{split}
\end{equation}
In the following we discuss how the various factors in the RHS of~\eqref{eq:partial_joint_after} become when inserting these transformed arguments.
% when using (\ref{eq:transformation_inverse}),

From the rotation-invariance of the Gaussian prior, the product of the two Gaussian factors becomes
\begin{equation} \label{eq:identGauss} 
\begin{split}
&\mathcal{N}(\widetilde{\bm{x}}\cos (\widetilde{\bm{\varphi}}_{\widetilde{k}} - \widetilde{\alpha})+\widetilde{\bm{\nu}}\sin(\widetilde{\bm{\varphi}}_{\widetilde{k}} - \widetilde{\alpha}); \bm{\mu}, \bm{\Sigma})
\cdot \mathcal{N}(-\widetilde{\bm{x}}\sin(\widetilde{\bm{\varphi}}_{\widetilde{k}} - \widetilde{\alpha})+\widetilde{\bm{\nu}}\cos(\widetilde{\bm{\varphi}}_{\widetilde{k}} - \widetilde{\alpha}); \bm{\mu}, \bm{\Sigma})\\ 
 &= \mathcal{N}(\widetilde{\bm{x}}; \bm{\mu}, \bm{\Sigma})\cdot \mathcal{N}(\widetilde{\bm{\nu}}; \bm{\mu}, \bm{\Sigma}),
\end{split}
\end{equation}
as mentioned in~\citet{Murray2010} and as can be easily verified by writing out the product of the Gaussian densities. 
\sloppy The next factors in (\ref{eq:joint_before}) to consider are $1/(r_i(\alpha, \bm{\varphi}_{1:i})-\ell_i(\alpha, \bm{\varphi}_{1:i}))$ for $i=0, \ldots, k-1$. Inserting the transformed arguments, the denominator in this expression becomes
\begin{align}\label{eq:uniformDensity}
    r_i(\widetilde{\varphi}_{\widetilde{k}\widetilde{m}}, \widetilde{\bm{\varphi}}_{0:i}) - \ell_i(\widetilde{\varphi}_{\widetilde{k}\widetilde{m}}, \widetilde{\bm{\varphi}}_{0:i}) = 
     r_i(\widetilde{\alpha}, \widetilde{\bm{\varphi}}_{0:i}) - \ell_i(\widetilde{\alpha}, \widetilde{\bm{\varphi}}_{0:i}),
\end{align}
where the equality follows from the flipping properties of Lemma~\ref{lem:recursive_mess}, which applies to the intervals and angles defined by the transformation because the  assumptions in Lemma~\ref{lem:recursive_mess} are fulfilled by the transformed support $\mathbb{I}\{T^{-1}(\widetilde{\bm{x}}, \widetilde{\bm{\phi}}, \widetilde{m})\in {\cal F}\}$. 
We then use the result in Lemma 2.1 to reformulate the indicator function, and specify what $P_{rs}(\bm{\psi}^{\uparrow})$ in~\eqref{eq:joint_before} becomes under the transformation.

As mentioned in the main text, to show invariance we need to sum over $m$ in~\eqref{eq:joint_before} and similarly over $\widetilde{m}$ in~\ref{eq:joint_after}. Defining $\mathcal{F}_{-m}$ as the set containing the support for $\bm{x}, \bm{\phi}$ only, it is possible to split the indicator function as a product of indicators, i.e.
\begin{equation}
    \mathbb{I}\{(\bm{x}, \bm{\phi}, m)\in \mathcal{F}\} = \mathbb{I}\{(\bm{x}, \bm{\phi})\in \mathcal{F}_{-m}\} \mathbb{I}\{m \in \mathcal{A}_k\},
\end{equation}
and analogously for the transformed variables. 
Then the sum for the forward density becomes
\begin{equation*}
\begin{split}
    \sum_{m} g_{\bm{x},\bm{\phi}}(\bm{x},\bm{\phi}) P_{rs}(\bm{\psi}^{\uparrow}) \mathbb{I}\{(\bm{x}, \bm{\phi}, m)\in \mathcal{F}\} 
    &= g_{\bm{x},\bm{\phi}}(\bm{x},\bm{\phi}) \mathbb{I}\{(\bm{x}, \bm{\phi})\in \mathcal{F}_{-m}\} \sum_{m} P_{h(\alpha; \bm{\psi}), h(\varphi_{km}; \bm{\psi})}(\bm{\psi}^{\uparrow})  \mathbb{I}\{m \in \mathcal{A}_k\} \\
    &= g_{\bm{x},\bm{\phi}}(\bm{x},\bm{\phi}) \mathbb{I}\{(\bm{x}, \bm{\phi})\in \mathcal{F}_{-m}\} \sum_{m \in \mathcal{A}_k} P_{h(\alpha; \bm{\psi}), h(\varphi_{km}; \bm{\psi})}(\bm{\psi}^{\uparrow})\\ 
    &=g_{\bm{x},\bm{\phi}}(\bm{x},\bm{\phi}) \mathbb{I}\{(\bm{x}, \bm{\phi})\in \mathcal{F}_{-m}\},
\end{split}
\end{equation*}
where the sum over all possible $m$ implies a summing the elements of row $r$ of the matrix. 
In the reverse path, under the transformation
\begin{equation*}
\begin{split}
    \sum_{\widetilde{m}} p_{\bm{x},\bm{\phi}}(\widetilde{\bm{x}}, \widetilde{\bm{\phi}}) P_{sr}(\bm{\psi}^{\uparrow}) \mathbb{I}\{(\widetilde{\bm{x}}, \widetilde{\bm{\phi}}, \widetilde{m})\in \mathcal{F}\}
    &= p_{\bm{x},\bm{\phi}}(\widetilde{\bm{x}}, \widetilde{\bm{\phi}}) \mathbb{I}\{(\widetilde{\bm{x}}, \widetilde{\bm{\phi}})\in \mathcal{F}_{-m}\} \sum_{\widetilde{m}} P_{h(\widetilde{\varphi}_{\widetilde{k}\widetilde{m}}; \widetilde{\bm{\psi}}), h(\widetilde{\alpha}; \widetilde{\bm{\psi}})}(\widetilde{\bm{\psi}}^{\uparrow}) \mathbb{I}\{\widetilde{m} \in \widetilde{\mathcal{A}}_{\widetilde{k}}\}\\
    &= p_{\bm{x},\bm{\phi}}(\widetilde{\bm{x}}, \widetilde{\bm{\phi}}) \mathbb{I}\{(\widetilde{\bm{x}}, \widetilde{\bm{\phi}})\in \mathcal{F}_{-m}\} \sum_{\widetilde{m} \in \widetilde{\mathcal{A}}_{\widetilde{k}}} P_{h(\widetilde{\varphi}_{\widetilde{k}\widetilde{m}}; \widetilde{\bm{\psi}}), h(\widetilde{\alpha}; \widetilde{\bm{\psi}})}(\widetilde{\bm{\psi}}^{\uparrow}) \\
    &= p_{\bm{x},\bm{\phi}}(\widetilde{\bm{x}}, \widetilde{\bm{\phi}}) \mathbb{I}\{(\widetilde{\bm{x}}, \widetilde{\bm{\phi}})\in \mathcal{F}_{-m}\} \sum_{\widetilde{m} \in \widetilde{\mathcal{A}}_{\widetilde{k}}} P_{h(\widetilde{\varphi}_{\widetilde{k}\widetilde{m}}; \widetilde{\bm{\psi}}), h(\widetilde{\alpha}; \widetilde{\bm{\psi}})}(\bm{\psi}^{\uparrow}) \\
    &=p_{\bm{x},\bm{\phi}}(\widetilde{\bm{x}}, \widetilde{\bm{\phi}}) \mathbb{I}\{(\widetilde{\bm{x}}, \widetilde{\bm{\phi}})\in \mathcal{F}_{-m}\},
\end{split}
\end{equation*}
where the equality in the third line follows from Remark~\ref{rem:psi_transformed}, and therefore the sum over $\widetilde{m}$ implies summing all elements in column $r$ of the matrix $\bm{P}(\bm{\psi}^{\uparrow})$, and is equal to one from the double-stochastic property. 
As a result, $(\bm{x},\bm{\nu},y,k, \alpha, \bm{\varphi}_1,\ldots,\bm{\varphi}_k)$ and $(\widetilde{\bm{x}},\widetilde{\bm{\nu}},\widetilde{y},\widetilde{k}, \widetilde{\alpha}, \widetilde{\bm{\varphi}}_1,\ldots,\widetilde{\bm{\varphi}}_{\widetilde{k}})$ are identically distributed, and thereby also the marginal distributions of $\bm{x}$ and $\widetilde{\bm{x}}$ are identical.

\section{Extended analysis of the toy solute transport example} \label{ap:toy_ad}

Consider the matrix model, approximation of a PDE for damped transport in fluids,
\begin{equation} \label{eq:ad_eq}
    \frac{d}{dt} \bm{\Theta} + (\bm{A} + \kappa \bm{I}) \bm{\Theta} = \bm{g}
\end{equation}
where the solution vector $\bm{\Theta} \in \mathbb{R}^d$ is the state variable (e.g. temperature, solute concentration), $\bm{A} \in \mathbb{M}^{[d\times d]}$ is a discretized version of the advection (transport) operator, $\kappa>0$ is a damping parameter, and $\bm{g} \in \mathbb{R}^d$ is a time-independent forcing term (e.g. heat source, dye injection or stirring). 
Last we choose $\bm{g} = \hat{\bm{e}}_1 = (1, 0, \ldots, 0)$ so that the energy injection is done only at mode 1 (the largest scale) and is cascaded to the other modes (smaller scales) by $\bm{A}$. 
 
We focus on the steady state of equation~\eqref{eq:ad_eq}, which satisfies the following
\begin{equation}\label{eq:ad_toy_sm}
    (\bm{A} + \kappa \bm{I}) \bm{\Theta} = \bm{g},
\end{equation}
and our goal is to use Bayesian inversion on~\eqref{eq:ad_toy_sm} to estimate the approximated advection operator $\bm{A}$.
Considering the partial observation model in~\eqref{eq:ad_obs_model},
where $P:\mathbb{R}^d \to \mathbb{R}^k$ is a projection operation so that $P(\bm{\Theta})=(\theta_{d_{0}-k}, \ldots, \theta_{d_{0}})$, we phrase the estimation problem as sampling from the posterior with measure
\begin{equation}
    \mu(d\bm{A}) \propto \exp(-\Phi(\bm{A}))\mu_0(d\bm{A}),
\end{equation}
with potential $\Phi(\bm{A}):= \sigma^{-2}|\bm{y}-P(\bm{\Theta}(\bm{A}))|^2/2$ that involves computing the solution vector
\begin{equation}
    \bm{\Theta}(\bm{A})=(\bm{A}+\kappa \bm{I})^{-1} \bm{g}
\end{equation}
at each iteration, what implies an expensive likelihood evaluation for large $d$. 

We choose prior measure  $\mu_0 := \mathcal{N}(\cdot; \bm{0}, \bm{C})$ with diagonal covariance operator $\bm{C} := \tau^2 \mbox{diag}(\bm{q})$, where $\bm{q} = (q_{ij})$ is a vector containing the variances of the elements $a_{ij}, \; 0 \leq i < j \leq d-1$ in the non-zero upper triangle of $\bm{A}$. More precisely, we define the elements of $\bm{q}$ as 
\begin{equation} \label{eq:ad_covariance}
    q_{ij} = \mbox{Var}(a_{ij}) := (ij)^{-\alpha}|i-j|^{-\gamma}, \quad 0 \leq i < j \leq d-1, \quad \gamma, \alpha > 0.
\end{equation}
With the modeling choice in~\eqref{eq:ad_covariance}, the magnitude of each element $a_{ij}$ will decay both with the difference between $i$ and $j$, discouraging long-range energy interactions, and with increasing indices $i$ and $j$, thereby penalizing energy transfer at higher frequencies. 

For this example, we set $\kappa=0.02$, $\alpha=3$, $\gamma=2$, $\sigma^2=0.25$, $\tau^2=2$, and generate datasets for $d \in \{10, 15, \ldots, 50\}$ using the model, as described next. 
Since our goal is to analyze the performance of the algorithms on the same dataset, but with higher resolution, we first generate $\bm{A}_{100}$, $\bm{z}\sim N_d(\bm{0}, \bm{I}_{100})$, $\bm{\varepsilon}\sim N_d(\bm{0}, \sigma_d^2\bm{I}_{100})$, and then build each $\bm{A}_d$,$\bm{\Theta}(\bm{A}_d)$, and the corresponding $\bm{y}_{d}$ by subsetting the upper $d\times d$ submatrix of $\bm{A}_{100}$ and the first $d$ elements of $\bm{z}_{100}$ and $\bm{\varepsilon}_{100}$. 
In addition, we set the observational scale with $d_0=6$ and $k=2$ for all $d$, corresponding to observing modes 4, 5, and 6 (the 4th, 5th and 6th elements of each $\bm{y}_d$). Figure~\ref{fig:ad_simulated_data} shows the generated $\bm{A}_d$,$\bm{\Theta}(\bm{A}_d)$, and $\bm{y}_d$ for $d\in\{10, 20, 40\}$.

\begin{figure}
    \centering
    \includegraphics[width=0.8\linewidth]{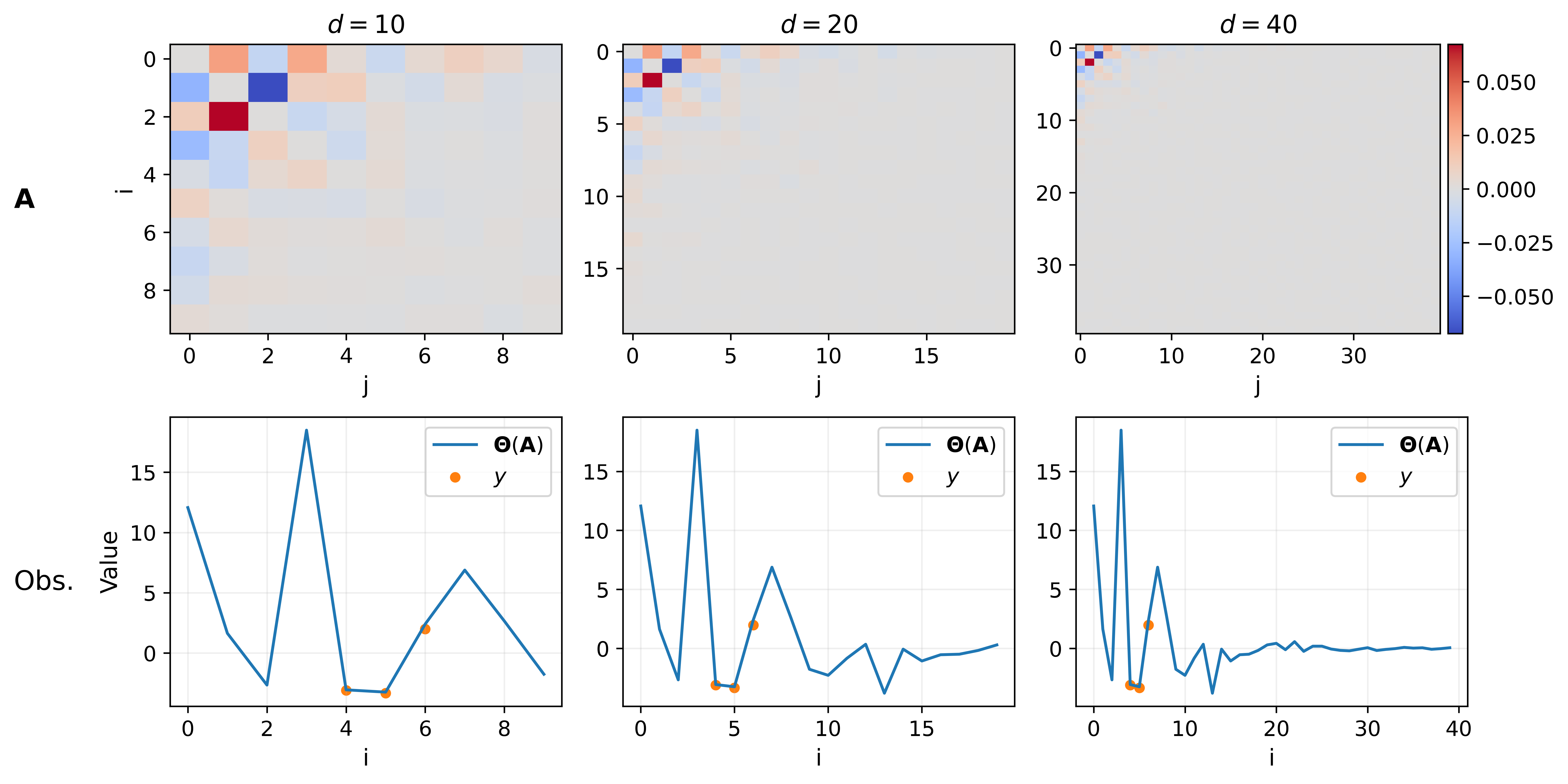}
    \caption{True simulated $\bm{A}_d$ and corresponding $\bm{\Theta}(\bm{A}_d)$ and observations $\bm{y}_{d}$ for dimensions $d\in\{10, 20, 40\}$.}
    \label{fig:ad_simulated_data}
\end{figure}

We then estimate the posteriors for each $d$ using MESS with $M\in[1, 10, 50, 100]$ and Metropolis-Hastings chains. To have a fair comparison with the prior-informed MESS, the MH proposal is a Gaussian centered around the current state with covariance matrix given as the product between a variance $\sigma_{mh}^2$ and the prior covariance $\bm{C}$. We tune the MH variance once, to give an acceptance rate of 23.4\% for $d=20$, and keep this value constant across the MH chains for every $d$. Fixing the MH proposal for all $d$ and fixing the observational scale as described in the paragraph above should ensure that the mixing of the chains is not compensated for, for increased dimension, by a better tuning or increased information. 

Figure~\ref{fig:ad_visualization} shows histograms of samples obtained with MESS with $M=100$ for $d=10$ and $d=50$. The histograms represent the marginal and joint pair-wise marginal distributions for elements $a_{01}-a_{04}$ in the first row 1 of $\bm{A}$, elements $a_{12}, a_{13}$ in the second row, and element $a_{23}$ in the third row. The elements $a_{01}, a_{12}$ and $a_{23}$ are contiguous to the main diagonal. The posterior exhibits a complicated geometry, exemplified by different kinds of multimodality and correlation structures observed in the histograms and pairwise densities. We note that the posterior geometry depends on the values chosen for the hyperparameters, the observational scale governed by $d_0$ and $k$, and of course the random seed used to create $\bm{z}_{100}$.  
\begin{figure}[t]
    \centering
    \begin{subfigure}{0.48\linewidth}
        \centering
        \includegraphics[width=\linewidth, trim={0.8cm 0 0 0}, clip]
        {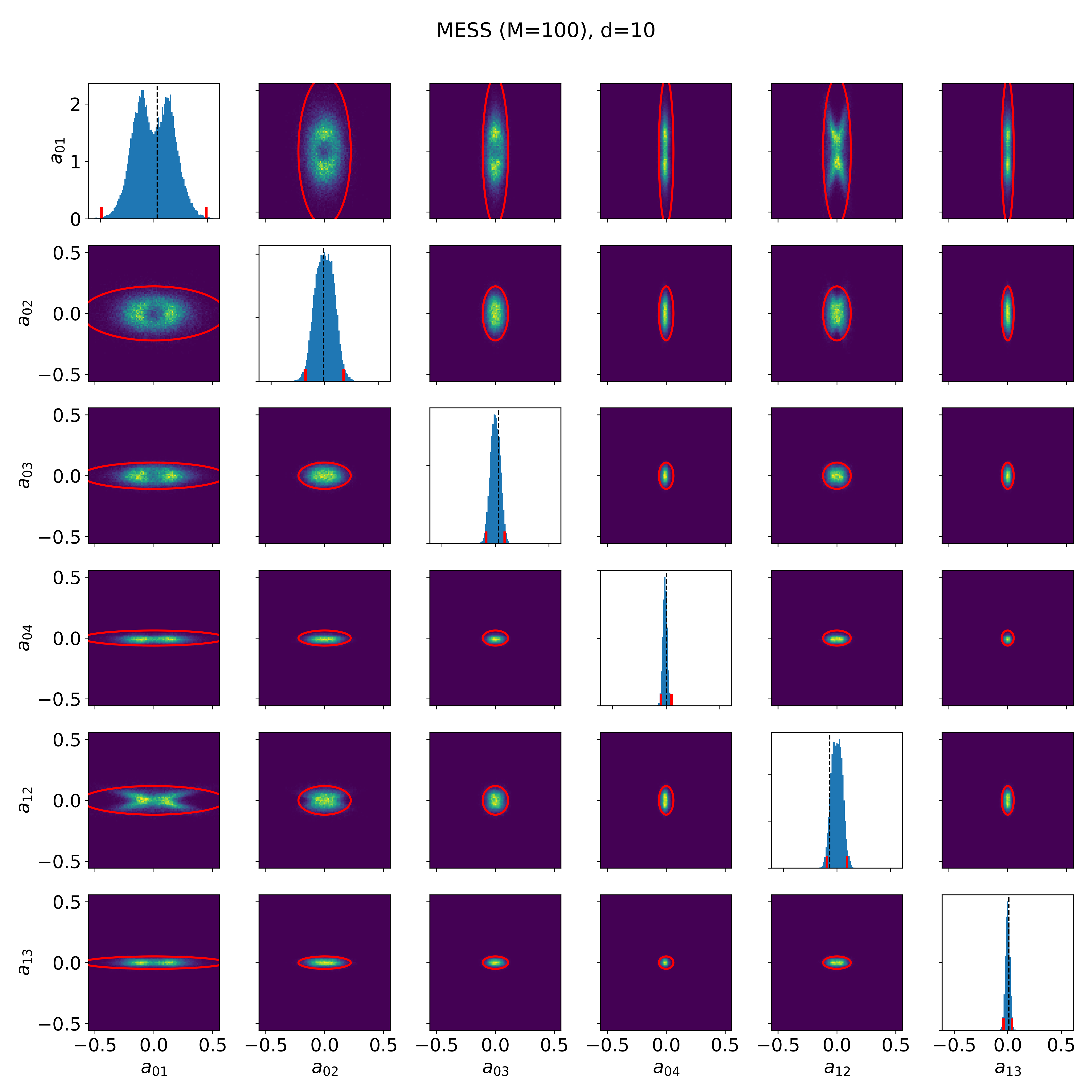}
    \end{subfigure}
    \hfill
    \begin{subfigure}{0.48\linewidth}
        \centering
        \includegraphics[width=\linewidth, trim={0.8cm 0 0 0}, clip]
        {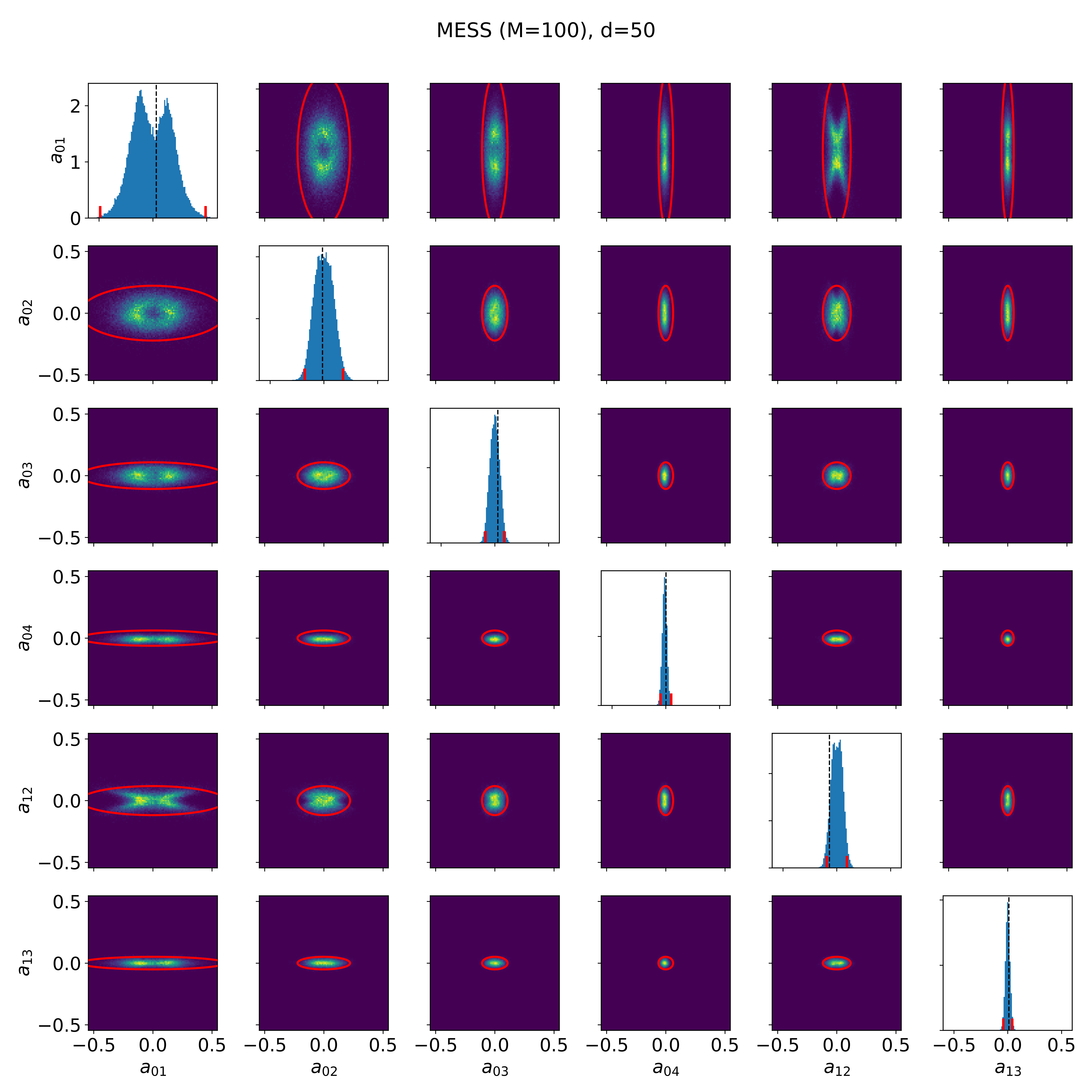}
    \end{subfigure}
    \caption{Posterior marginals and pair-wise densities for the same selected components of $\bm{A}_d$, for $d=10$ and $d=50$, illustrating that the complicated posterior geometry is preserved across dimensions. The samples were obtained with MESS ($M=100$).}
    \label{fig:ad_visualization}
\end{figure}

Figure~\ref{fig:ad_traceplots_a01} displays traceplots for component $a_{01}$ of $\bm{A}$. In the figure, each row is a value for $d$, and each column is an algorithm. The traceplots illustrate the different chain mixing behaviour: MESS ($M=50$) mixes better than MESS ($M=1$), and both mix better than Metropolis-Hasting (MH). In addition, the mixing does not degrade with dimension for both MESS variants, while it does so with MH. 

\begin{figure}
    \centering
    \includegraphics[width=1\linewidth]{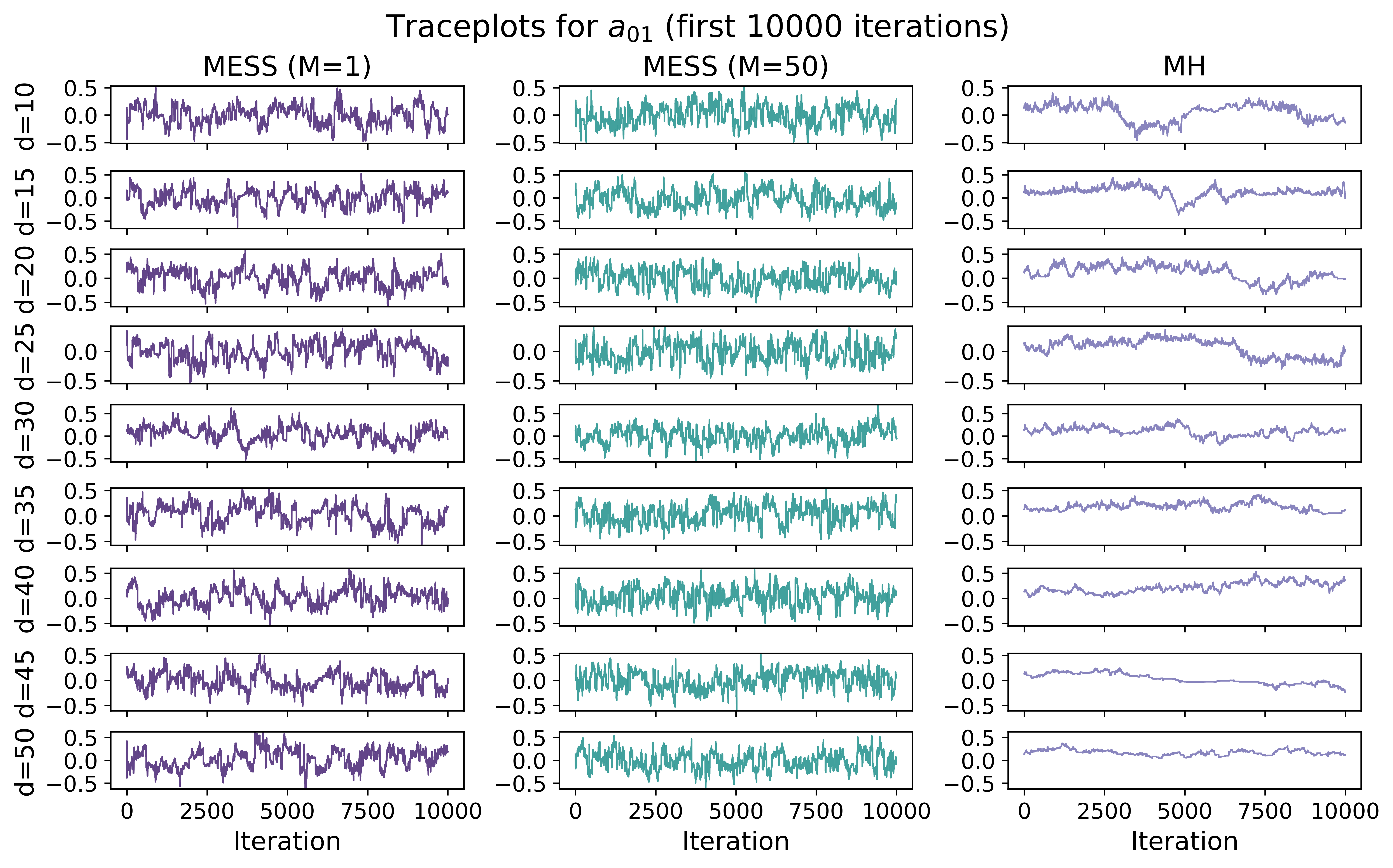}
    \caption{Traceplots comparing the mixing behaviour of different algorithms for increasing $d$, for the first element $a_{01}$ of the advection matrix $\bm{A}$.}
    \label{fig:ad_traceplots_a01}
\end{figure}

Finally, Table~\ref{tab:ess_msjd_ad_a02} contains the mean effective sample size and MSJD for component $a_{02}$ of $\bm{A}$ (see Section~\ref{sec:results} in the main text for the similar table for $a_{01}$), illustrating that for this component informing the transition matrix with angular or Euclidean distance improved the mixing in about 20\% for $M=50$. 
\begin{table} 
\centering
\small
\begin{tabular}{l|ccc|ccc}
\toprule
& \multicolumn{3}{c|}{Mean Eff. Sample Size} & \multicolumn{3}{c}{MSJD} \\
$M$ & Unif & Ang & Eucl & Unif & Ang & Eucl \\
\midrule
10 & 2094 & \textbf{2318} & 2077 & 0.00034 & 0.00031 & 0.00032 \\
50 & 2545 & 3095 & \textbf{3187} & 0.00046 & 0.00045 & 0.00046 \\
\bottomrule
\end{tabular}
\caption{ESS/MSJD for component $a_{02}$ at $d=10$ using the first 300k samples.}
\label{tab:ess_msjd_ad_a02}
\end{table}

\end{document}